\UseRawInputEncoding
\documentclass[journal,final,twocolumn]{IEEEtran}
\usepackage{fixltx2e}
\usepackage{cite}
\usepackage{mathrsfs}
\usepackage{color}
\usepackage{float}
\usepackage[caption=false]{subfig}

\ifCLASSINFOpdf
   \usepackage[pdftex]{graphicx}
   \graphicspath{{Figs/}}
   \DeclareGraphicsExtensions{.pdf,.jpeg,.png}
\else
\fi

\usepackage[cmex10]{amsmath}
\usepackage{amsmath}
\usepackage{amssymb}
\usepackage{amsthm}
\usepackage{amsfonts}
\usepackage{bm}
\usepackage{xfrac}
\usepackage{empheq}
\usepackage[normalem]{ulem} 
\usepackage{soul} 
\usepackage{mathtools}

\newtheorem{theorem}{Theorem}

\newtheorem{example}{Example}
\newtheorem{proposition}{Proposition}
\newtheorem{lemma}{Lemma}

\newtheorem{corollary}{Corollary}
\newtheorem{remark}{Remark}
\theoremstyle{definition}
\newtheorem{definition}{Definition}
\usepackage{xcolor}
\usepackage{color}

\graphicspath{{figs/}}

\begin{document}
	
	\title{Data Disclosure with Non-zero Leakage and Non-invertible Leakage Matrix}
\vspace{-5mm}
\author{
		\IEEEauthorblockN{Amirreza Zamani, Tobias J. Oechtering, Mikael Skoglund \vspace*{0.5em}
			\IEEEauthorblockA{\\
                              Division of Information Science and Engineering, KTH Royal Institute of Technology \\
				Email: \protect amizam@kth.se, oech@kth.se, mikael.skoglund@ee.kth.se }}
\thanks{This work was funded in
	part by the Swedish research council under contract 2019-03606. A. Zamani , M. Skoglund and T. J. Oechtering are with the division of information science and engineering, School of Electrical Engineering and
	Computer Science, KTH Royal Institute of Technology, 100 44 Stockholm,
	Sweden (e-mail: amizam@kth.se; oech@kth.se; mikael.skoglund@ee.kth.se).}}
%
	\maketitle

\begin{abstract}
	We study a statistical signal processing privacy problem, where an agent observes useful data $Y$ and wants to reveal the information to a user. Since the useful data is correlated with the private data $X$, the agent employs a privacy mechanism to generate data $U$ that can be released. We study the privacy mechanism design that maximizes the revealed information about $Y$ while satisfying a strong $\ell_1$-privacy criterion.
	When a sufficiently small leakage is allowed, we show that the optimizer vectors of the privacy mechanism design problem have a specific geometry, i.e., they are perturbations of fixed vector distributions. This geometrical structure allows us to use a local approximation of the conditional entropy. By using this approximation the original optimization problem can be reduced to a linear program so that an approximate solution for privacy mechanism can be easily obtained. 
	The main contribution of this work is to consider non-zero leakage with a non-invertible leakage matrix.
	In an example inspired by water mark application, we first investigate the accuracy of the approximation. Then, we employ different measures for utility and privacy leakage to compare the privacy-utility trade-off using our approach with other methods. 
	In particular, it has been shown that by allowing small leakage, significant utility can be achieved using our method compared to the case where no leakage is allowed.    
	
	
\end{abstract}

\section{Introduction}
The amount of data generated by software system, interconnected sensors that record and process signals from the physical world, robots and humans is growing rapidly. Direct disclosure of raw data can cause privacy breaches through illegitimate inferences. Therefore, privacy mechanisms are needed to control the disclosure of the data.

Privacy mechanism design from a statistical signal processing perspective is receiving increasing attention recently and related results can be found in \cite{yamamoto, sankar, makhdoumi, dwork1, oech, issa, issa3, khodam, Calmon2, asoodeh1, houi, borz,Total, gun, Wyner, courtade, sankar2, deniz4, asoodeh3, Calmon1, 7888175, nekouei2, Johnson, aamir,jende}.
Specifically, the concept of maximal leakage has been introduced in \cite{issa} and used in \cite{issa2} for the Shannon cipher system. Furthermore, some bounds on the privacy utility trade-off are derived. In \cite{makhdoumi}, the concept of a privacy funnel is introduced where the privacy utility trade-off considering log-loss as privacy metric and distortion metric for utility has been studied. The concept of differential privacy is introduced in \cite{dwork1, dwork2}, which aims to answer queries in a privacy preserving approach using minimizing the chance of identifying the membership in an statistical database.
 In \cite{oech}, the hypothesis test performance of an adversary is used to measure the privacy leakage.
Fundamental limits of privacy utility trade-off are studied in \cite{Calmon2}, measuring the leakage using estimation-theoretic guarantees.
A source coding problem with secrecy is studied in \cite{yamamoto}.
In both \cite{yamamoto} and \cite{sankar}, the privacy-utility trade-off is considered using expected distortion as a measure of utility and equivocation as measure of privacy.
 In \cite{asoodeh1}, maximal correlation either mutual information are used for measuring the privacy and properties of rate-privacy functions are studied.
 \begin{figure}[]
 	\centering
 	\includegraphics[width = 0.4\textwidth]{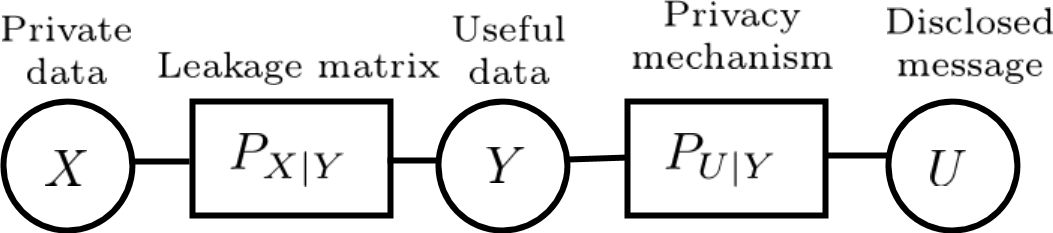}
 	\caption{In this system model, disclosed data $U$ is designed by a privacy mechanism which maximizes the information disclosed about $Y$ and satisfies the strong $l_1$-privacy criterion.}
 	\label{fig:sysmodel1}
 \end{figure}
In this paper, an agent tries to reveal some useful information to a user as shown in Fig.~\ref{fig:sysmodel1}. Random variable (rv) $Y$ denotes the useful data and is dependent on private data denoted by rv $X$, which may not be accessible by the agent. The useful data can not be disclosed directly since the agent wants to protect the private data. Thus, the agent employs a privacy mechanism to design disclosed data denoted by rv $U$. $U$ should disclose as much information about $Y$ as possible and fulfill the privacy criterion simultaneously. Here we extend \cite{borz} and our previous work \cite{khodam} dealing with an open issue, i.e., considering non-zero leakage and non-invertible leakage matrix. In this work we consider a point-wise $\ell_1$-privacy criterion instead of point-wise $\chi^2$-privacy criterion considered in \cite{khodam}, which is a small variation that will be discussed. 
In the conference version \cite{aamir}, an initial brief study of the privacy mechanism design has been presented. Here, a complete study is presented that includes the studies of the geometrical properties of optimizing vectors, range of permissible leakage, numerical examples to illustrate the design, a numerical example that compares our approximate solution with the exact solution and finally numerical studies where employ different measures for privacy and utility to evaluate further our method.
In \cite{borz}, the problem of maximizing utility, i.e., $I(U;Y)$, under the leakage constraint $I(U;X)\leq \epsilon$ and Markov chain $X-Y-U$ is studied and it is shown that under perfect privacy assumption, i.e., $\epsilon=0$, the privacy mechanism design can be reduced to standard linear programming. This work is extended in \cite{gun}, considering the privacy utility trade-off with a rate constraint on released data. 
Next, in \cite{khodam}, the problem of maximizing $I(U;Y)$ under a strong $\chi^2$-privacy criterion and Markov chain $X-Y-U$ is studied. It has been shown that for small $\epsilon$, the original highly challenging optimization problem can be approximated by a problem of finding the principal right-singular vector of a matrix using the concept of Euclidean information geometry. Furthermore, a geometrical interpretation of the privacy mechanism design has been provided. 

 In this work, we consider an element-wise $\ell_1$-privacy criterion which we call \emph{"strong $\ell_1$-privacy criterion"}. In \cite{Total}, an average total variation is used as privacy measure, which can be obtained by taking the average of the strong $\ell_1$-privacy criterion. A $\chi^2$-privacy criterion is considered in \cite{Calmon2}, where an upper bound and a lower bound on the privacy utility trade-off have been derived. Likewise, the strong $\chi^2$-privacy criterion used in \cite{khodam}, is a point-wise criterion, while, in \cite{Calmon2}, an average criterion is employed. 
In \cite{khodam}, the leakage matrix, i.e, $P_{X|Y}$, is assumed to be invertible, however, in this work we generalize the scenario by assuming $P_{X|Y}$ to be a full row rank, also known as \emph{fat matrix}.

In order to simplify the design problem, we use similar concepts as used in \cite{borade, huang}.
In more detail we use methods from Euclidean information theory to approximate KL divergence and mutual information in order to simplify the optimization problem. More specifically,
we show that the optimizing distributions have special geometry and therefore we can approximate the mutual information using information geometry concepts. We show that the main problem can be approximated by a standard linear programming when small leakage is allowed. Furthermore, when the leakage tends to zero we obtain the same linear programming as in \cite{borz}. 

 Our contribution can be summarized as follows:\\
\textbf{(i)} In Section \ref{sec:system}, we present a statistical signal processing privacy problem using a strong $\ell_1$-privacy criterion.\\
\textbf{(ii)} in Section \ref{result}, we show that the distributions $P_{Y|U=u}$ can be decomposed into two parts, where the first part refers to the null space of $P_{X|Y}$ and the second part refers to the non-singular part of $P_{X|Y}$. Then, by using this decomposition we characterize the optimizing vectors of the main problem, which have a special geometry. By utilizing concepts from information geometry, we then approximate the main problem by a standard linear programming.\\
\textbf{(iii)} In Section \ref{disc}, we first illustrate the geometrical properties of the optimizer vectors, we then investigate the range of permissible leakage and derive upper bounds on the leakage.\\ 
\textbf{(iv)} In Section \ref{example}, we present an example inspired by water mark application. First, we summarize our mechanism design in a numerical example, we then compare the approximate solution with the optimal solution found by exhaustive search. Later, we employ different metrics for utility and privacy to illustrate the privacy-utility trade-off. For measuring utility we use probability of error using MAP decision making rule and normalized MMSE$(Y|U)$, for privacy leakage measure we employ normalized MMSE($X|U$). It is shown that by allowing a small information leakage, we can achieve significant utility compared to the case where we have $\epsilon=0$.\\
The paper is concluded in Section \ref{concul1}.
\section{system model and Problem Formulation} \label{sec:system}
Let $P_{XY}$ denote the joint distribution of discrete random variables $X$ and $Y$ defined on the finite alphabets $\cal{X}$ and $\cal{Y}$ with $|\cal{X}|<|\cal{Y}|$.
 We represent $P_{XY}$ by a matrix defined on $\mathbb{R}^{|\mathcal{X}|\times|\mathcal{Y}|}$ and 
marginal distributions of $X$ and $Y$ by vectors $P_X$ and $P_Y$ defined on $\mathbb{R}^{|\mathcal{X}|}$ and $\mathbb{R}^{|\mathcal{Y}|}$ given by row and column sums of $P_{XY}$. We assume that each element in vectors $P_X$ and $P_Y$ is nonzero. Furthermore, we represent the leakage matrix $P_{X|Y}$ by a matrix defined on $\mathbb{R}^{|\mathcal{X}|\times|\cal{Y}|}$, which is assumed to be of full row rank. Thus, without loss of generality we assume that $P_{X|Y}$ can be represented by two submatrices where the first submatrix is invertible, i.e., $P_{X|Y}=[P_{X|Y_1} , P_{X|Y_2}]$ such that $P_{X|Y_1}$ defined on $\mathbb{R}^{|\mathcal{X}|\times|\mathcal{X}|}$ is invertible.
In the privacy problem with perfect privacy \cite{borz}, it has been shown that positive amount of information can be only revealed if $P_{X|Y}$ is not invertible. This result was also proved in \cite{berger} in a source coding setup.
In \cite{khodam}, we designed a privacy mechanism for revealing information for an invertible $P_{X|Y}$, where a small leakage is allowed. The present work provides a framework where these two results are unified for a full row rank leakage matrix. 
 RVs $X$ and $Y$ denote the private data and the useful data. 
 
 In this work, privacy is measured by the \emph{strong $\ell_1$-privacy criterion} which we introduce next. It is called strong since it is an element-wise criterion and the used norm is the $\ell_1$ norm. 
 \begin{definition}
 	For private data $X\in\mathcal{X}$, the disclosed data $U\in\mathcal{U}$ satisfies the \emph{strong $\ell_1$-privacy criterion} with leakage $\epsilon\geq0$ if the joint pmf $P_{XU}$ satisfies for all $u\in\mathcal{U}$ the following condition
 	 \begin{align*}
 	& \left\lVert P_{X|U=u}-P_X \right\rVert_1\!=\!\sum_x|P_{X|U=u}(x)-P_X(x)|\leq \epsilon,\ \forall u\in\mathcal{U}.
 	\end{align*}
 \end{definition}
  Intuitively, for small $\epsilon$, the privacy criterion means that the two distributions (vectors) $P_{X|U=u}$ and $P_X$ are close to each other. This should hold for all $u\in\mathcal{U}$. Thus $X$ and $U$ are almost independent in the sense that $P_{X|U=u}$ almost does not depend on $U$.
The closeness of $P_{X|U=u}$ and $P_X$ allows us to locally approximate the conditional entropy $H(Y|U)$.
In the next proposition, the relation between the strong $\ell_1$-privacy criterion and the $\chi^2$-strong privacy criterion defined in \cite[Definition 1]{khodam} is studied.
\begin{proposition}\label{ineq}
	The $\ell_1$-privacy criterion implies the strong $\chi^2$-privacy criterion\footnote{$\chi^2(P_{X|U=u}||P_X)=
		\sum_{x\in\mathcal{X}}\frac{(P_{X|U=u}(x)-P_X(x))^2}{P_X(x)}\leq\epsilon^2,\ \forall u\in\mathcal{U}.$} with different $\epsilon$, i.e, if $\left\lVert P_{X|U=u}-P_X \right\rVert_1\leq \epsilon$, then we have $\chi^2(P_{X|U=u}||P_X)\leq (\epsilon')^2$, where $\epsilon'=\frac{\epsilon}{\sqrt{\min P_X}}$. For the other direction, if $\chi^2(P_{X|U=u}||P_X)\leq \epsilon^2$, then $\left\lVert P_{X|U=u}-P_X \right\rVert_1\leq \epsilon$.
\end{proposition}
\begin{proof}
	The first direction holds since 
	\begin{align*}
	\chi^2(P_{X|U=u}||P_X)&=
	\left\lVert [\sqrt{P_X}^{-1}]\left(P_{X|U=u}-P_X\right)\right\rVert_2^2\\&\leq\frac{1}{\min P_X}\left\lVert P_{X|U=u}-P_X\right\rVert_2^2\stackrel{(a)}{\leq} (\epsilon')^2,
	\end{align*}
	where (a) comes from the fact that the $\ell_2$-norm is upper-bounded by $\ell_1$-norm. The other direction is true since we have $\chi^2(P_{X|U=u}||P_X)\geq 4\text{TV}(P_{X|U=u}||P_X)^2$, where TV$(.||.)$ is the total variation distance.
\end{proof}
\begin{remark}
	As shown in \cite{khodam} by using an inequality between KL-divergence and $\chi^2$-privacy criterion, the strong $\ell_1$-privacy criterion implies bounded mutual information, i.e., $I(U,X)\leq(\epsilon')^2$. \footnote{Furthermore, by using the concept of Euclidean information theory $(\epsilon')^2$ can be strengthen and replaced by $\frac{1}{2}(\epsilon')^2+o((\epsilon')^2)$ for small $\epsilon'$ \cite[Proposition 2]{khodam}. Thus, \eqref{local} implies $I(X;U)\leq\frac{1}{2}(\epsilon')^2+o((\epsilon')^2)$ for small $\epsilon'$.}
\end{remark}
\begin{remark}
	If we take average of the strong $\ell_1$-privacy criterion over $U$, we get the bounded $f$-information for $f(x)=|x-1|$ similarly as $\chi^2$-information is the average of the privacy measure used in \cite{khodam}, see also \cite[Definition 7]{cal}.
\end{remark}
Our goal is to design a privacy mechanism that produces disclosed data $U$ that maximizes $I(U;Y)$ and satisfies the strong $\ell_1$-privacy criterion. The relation between $U$ and $Y$ is described by the kernel $P_{U|Y}$ defined on $\mathbb{R}^{|\mathcal{U}|\times|\mathcal{Y}|}$.
The privacy problem can be stated as follows
\begin{subequations}\label{prob1}
\begin{align}
\sup_{P_{U|Y}} \ \ &I(U;Y),\label{privacy}\\
\text{subject to:}\ \ &X-Y-U,\label{Markov}\\
& \left\lVert P_{X|U=u}-P_X \right\rVert_1\leq \epsilon,\ \forall u\in\mathcal{U}.\label{local}
\end{align}
\end{subequations}
\begin{remark}
	By letting $\epsilon=0$, \eqref{prob1} leads to the perfect privacy design problem studied in \cite{borz}.
\end{remark}
\begin{proposition}\label{prop1111}
	It suffices to consider $U$ such that $|\mathcal{U}|\leq|\mathcal{Y}|$. Since the supremum in \eqref{privacy} is achieved, we can replace the supremum by the maximum.
\end{proposition}
\begin{proof}
	The proof is provided in Appendix~A.
\end{proof}
In next corollary, by using Proposition~\ref{ineq} the relation between \eqref{prob1} and the case where $\chi^2$ measure is employed instead of $\ell_1$ is shown.
\begin{corollary}
	Let $g_{\epsilon}(X,Y)=\max\limits_{\substack{
		P_{U|Y}: X-Y-U\\\chi^2(P_{X|U=u}||P_X)\leq\epsilon^2,\ \forall u}} I(U;Y)$ and $f_{\epsilon}(X,Y)=\max\limits_{\substack{
		P_{U|Y}:X-Y-U\\\left\lVert P_{X|U=u}-P_X \right\rVert_1\leq\epsilon,\ \forall u}} I(U;Y)$. Then we have
	\begin{align*}
	g_{\epsilon}(X,Y)\leq f_{\epsilon}(X,Y) \leq g_{\epsilon'}(X,Y),
	\end{align*}
	where $\epsilon'=\frac{\epsilon}{\sqrt{\min P_X}}$.
\end{corollary}
\section{Privacy Mechanism Design}\label{result}
In this section we show that the problem defined in \eqref{privacy} can be approximated by a quadratic problem. Furthermore, it is shown that the quadratic problem can be converted to a standard linear program.
For a sufficiently small $\epsilon$ by using \eqref{local},  we can rewrite the conditional distribution $P_{X|U=u}$ as a perturbation of $P_X$. Thus, for any $u\in\mathcal{U}$, similarly as in \cite{khodam,borade, huang}, we can write $P_{X|U=u}=P_X+\epsilon\cdot J_u$, where $J_u\in\mathbb{R}^{|\mathcal{X}|}$ is a perturbation vector that has the following three properties:
\begin{align}
\sum_{x\in\mathcal{X}} J_u(x)=0,\ \forall u,\label{prop1}\\
\sum_{u\in\mathcal{U}} P_U(u)J_u(x)=0,\ \forall x\label{prop2},\\
\sum_{x\in\mathcal{X}}|J_u(x)|\leq 1,\ \forall u\label{prop3}.
\end{align} 
The first two properties ensure that $P_{X|U=u}$ is a valid probability distribution and the third property follows from \eqref{local}. 
In Appendix~B, three lemmas are provided which are used to rewrite the main problem defined in \eqref{prob1}. In Lemma~\ref{null}, two properties of the null space Null$(P_{X|Y})$ are shown. Later, in Lemma~\ref{null2}, we show that for every $u\in\mathcal{U}$, the vector $P_{Y|U=u}$ lies in a convex polytope $\mathbb{S}_u$ defined as follows: 
\begin{align}\label{Su}
\mathbb{S}_{u} \triangleq \left\{y\in\mathbb{R}^{|\mathcal{Y}|}|My=MP_Y+\epsilon M\begin{bmatrix}
P_{X|Y_1}^{-1}J_u\\0
\end{bmatrix},\ y\geq \bm{0}\right\},
\end{align}
where $\begin{bmatrix}
P_{X|Y_1}^{-1}J_u\\0
\end{bmatrix}\in\mathbb{R}^{|\cal Y|}$ and $J_u$ satisfies \eqref{prop1}, \eqref{prop2}, and \eqref{prop3}. Furthermore, $M\in\mathbb{R}^{|\mathcal{X}|\times|\cal Y|}$ is defined in Lemma~1. Using Lemma~\ref{null2}, it can be seen that for every $u$ the vectors $P_{Y|U=u}-P_{Y}$ can be decomposed into two parts, where the first part refers to Null$(P_{X|Y})$ and the second part refers to the non-singular part of $P_{X|Y}$. In other words, information can be revealed using singular part and non-singular part of $P_{X|Y}$. In \cite{borz}, it has been shown that information is revealed using only the null space Null$(P_{X|Y})$ and in \cite{khodam}, information is disclosed using the non-singular part of $P_{X|Y}$. In this work, we unify these two results.
\begin{theorem}
	We have the following equivalency 
	\begin{align}\label{equi}
	\min_{\begin{array}{c} 
		\substack{P_{U|Y}:X-Y-U\\ \|P_{X|U=u}-P_X\|_1\leq\epsilon,\ \forall u\in\mathcal{U}}
		\end{array}}\! \! \! \!\!\!\!\!\!\!\!\!\!\!\!\!\!\!\!H(Y|U) =\!\!\!\!\!\!\!\!\! \min_{\begin{array}{c} 
		\substack{P_U,\ P_{Y|U=u}\in\mathbb{S}_u,\ \forall u\in\mathcal{U},\\ \sum_u P_U(u)P_{Y|U=u}=P_Y,\\ J_u \text{satisfies}\ \eqref{prop1},\ \eqref{prop2},\ \text{and}\ \eqref{prop3}}
		\end{array}} \!\!\!\!\!\!\!\!\!\!\!\!\!\!\!\!\!\!\!H(Y|U).
	\end{align}
\end{theorem}
\begin{proof}
	The proof follows directly from Lemma~\ref{null2} and Lemma~\ref{44} provided in Appendix~B.
\end{proof}
In next proposition we discuss how  $H(Y|U)$ is minimized over $P_{Y|U=u}\in\mathbb{S}_u$ for all $u\in\mathcal{U}$.
\begin{proposition}\label{4}
		Let $P^*_{Y|U=u},\ \forall u\in\mathcal{U}$ be the minimizer of $H(Y|U)$ over the set $\{P_{Y|U=u}\in\mathbb{S}_u,\ \forall u\in\mathcal{U}|\sum_u P_U(u)P_{Y|U=u}=P_Y\}$, then 
	$P^*_{Y|U=u}\in\mathbb{S}_u$ for all $u\in\mathcal{U}$ must belong to extreme points of $\mathbb{S}_u$.
\end{proposition}
\begin{proof}
	The proof builds on the concavity property of the entropy. For more details see Appendix~C.
\end{proof}

In order to solve the minimization problem found in \eqref{equi}, we propose the following procedure:
In the first step, we find the extreme points of $\mathbb{S}_u$ denoted by $\mathbb{S}_u^*$, which is an easy task. Since the extreme points of the sets $\mathbb{S}_u$ have a particular geometry, in the second step, we locally approximate the conditional entropy $H(Y|U)$ so that we end up with a quadratic problem with quadratic constraints over $P_U(.)$ and $J_u$ that can be easily solved.
\subsection{Finding $\mathbb{S}_u^*$ (Extreme points of $\mathbb{S}_u$)}\label{seca}
In this part we find the extreme points of $\mathbb{S}_u$ for each $u\in\mathcal{U}$. 
As argued in \cite{borz}, the extreme points of $\mathbb{S}_u$ are the basic feasible solutions of $\mathbb{S}_u$. Basic feasible solutions of $\mathbb{S}_u$ can be found in the following manner. Let $\Omega$ be the set of indices which correspond to $|\mathcal{X}|$ linearly independent columns of $M$, i.e., $|\Omega|=|\mathcal{X}|$ and $\Omega\subset \{1,..,|\mathcal{Y}|\}$. Let $M_{\Omega}\in\mathbb{R}^{|\mathcal{X}|\times|\mathcal{X}|}$ be the submatrix of $M$ with columns indexed by the set $\Omega$. It can be seen that $M_{\Omega}$ is an invertible matrix since $rank(M)=|\cal X|$. Then, if all elements of the vector $M_{\Omega}^{-1}(MP_Y+\epsilon M\begin{bmatrix}
P_{X|Y_1}^{-1}J_u\\0\end{bmatrix})$ are non-negative, the vector $V_{\Omega}^*\in\mathbb{R}^{|\mathcal{Y}|}$, which is defined in the following, is a basic feasible solution of $\mathbb{S}_u$. 
Assume that $\Omega = \{\omega_1,..,\omega_{|\mathcal{X}|}\}$, where $\omega_i\in\{1,..,|\mathcal{Y}|\}$ and all elements are arranged in an increasing order. The $\omega_i$-th element of $V_{\Omega}^*$ is defined as $i$-th element of $M_{\Omega}^{-1}(MP_Y+\epsilon M\begin{bmatrix}
P_{X|Y_1}^{-1}J_u\\0\end{bmatrix})$, i.e., for $1\leq i \leq |\mathcal{X}|$ we have
\begin{align}\label{defin}
V_{\Omega}^*(\omega_i)= \left(M_{\Omega}^{-1}MP_Y+\epsilon M_{\Omega}^{-1}M\begin{bmatrix}
P_{X|Y_1}^{-1}J_u\\0\end{bmatrix}\right)(i).
\end{align}
Other elements of $V_{\Omega}^*$ are set to be zero. In the next proposition we show two properties of each vector inside $S_{u}^*$.
\begin{proposition}\label{kos}
	Let $\Omega\subset\{1,..,|\mathcal{Y}|\},\ |\Omega|=|\cal X|$. For every $\Omega$ we have $1^T\left(M_{\Omega}^{-1}MP_Y \right)=1$. Furthermore, $1^T\left(M_{\Omega}^{-1}M\begin{bmatrix}
	P_{X|Y_1}^{-1}J_u\\0\end{bmatrix}\right)=0$.
\end{proposition}
\begin{proof}
	The proof is provided in Appendix~D.
\end{proof}
Next we define a set $\mathcal{H}_{XY}$ which includes all joint distribution matrices $P_{XY}$ having a property that allows us to approximate the conditional entropy $H(Y|U)$.
\begin{definition}
	Let $\mathbb{R}_{+}$ and $\mathbb{R}_{++}$ denote the set of non-negative and positive reals. Let $\mathcal{P}(\cal X)$ be the standard probability simplex defined as
	$\mathcal{P}(\mathcal {X})=\{x\in\mathbb{R}^{|\mathcal{X}|}_{+}|1^Tx=1\}$
	and let $\mathcal{P'}(\cal X)$ be the subset of $\mathcal{P}(\cal X)$ defined as
	$\mathcal{P'}(\mathcal {X})=\{x\in\mathbb{R}^{|\mathcal{X}|}_{++}|1^Tx=1\}$.
	For every $\Omega\subset\{1,..,|\cal Y|\}$, $|\Omega|=|\mathcal{X}|$, $\mathcal{H}_{XY}$ is the set of all joint distributions $P_{XY}$ defined as follows
	\begin{align*}
	\mathcal{H}_{XY} = \{P_{XY}\!\in\! \mathbb{R}^{|\cal X|\times |\cal Y|}|\text{if}\ t_{\Omega}\in\mathcal{P}(\mathcal {X})\Rightarrow t_{\Omega}\in\mathcal{P'}(\mathcal {X})\},
	\end{align*}
	where $t_{\Omega}=M_{\Omega}^{-1}MP_Y$. In other words, for any $P_{XY}\in\mathcal{H}_{XY}$, if $t_{\Omega}$ is a probability vector, then it includes positive elements for every $\Omega$.
\end{definition}
In the following we provide an example to clarify the definition of $\mathcal{H}_{XY}$.
\begin{example}
	Let $P_{X|Y}=\begin{bmatrix}
	0.3   &0.8  &0.5 \\0.7 &0.2 &0.5 
	\end{bmatrix} $ and $P_Y=[\frac{2}{3},\frac{1}{6},\frac{1}{6}]^T$. By using SVD of $P_{X|Y}$, $M$ can be found as $\begin{bmatrix}
	-0.5556   &-0.6016  &0.574 \\0.6742 &-0.73 &0.1125 
	\end{bmatrix}$. Possible sets $\Omega$ are $\Omega_1=\{1,2\},\ \Omega_2=\{1,3\},$ and $\Omega_3=\{2,3\}$. By calculating $M_{\Omega_{u_i}}^{-1}MP_Y$ we have
	\begin{align*}
	&M_{\Omega_{1}}^{-1}MP_Y =\begin{bmatrix}
	0.7667\\0.2333
	\end{bmatrix},\ M_{\Omega_{2}}^{-1}MP_Y =\begin{bmatrix}
	0.4167\\0.5833
	\end{bmatrix},\\
	&M_{\Omega_{3}}^{-1}MP_Y =\begin{bmatrix}
	-0.2778\\1.2778
	\end{bmatrix}.
	\end{align*}
	Since the elements of the first two vectors are positive, $P_{XY}\in \mathcal{H}_{XY}$. On the other hand, let $P_{X|Y}=\begin{bmatrix}
	0.2   &0.1  &0.5 \\0.8 &0.9 &0.5 
	\end{bmatrix} $ and $P_Y=[\frac{1}{3},\frac{1}{2},\frac{1}{6}]^T$. We have
	\begin{align*}
	M_{\Omega_{1}}^{-1}MP_Y \!=\!\begin{bmatrix}
	1\\0
	\end{bmatrix},\ M_{\Omega_{2}}^{-1}MP_Y \!=\!\begin{bmatrix}
	1\\0
	\end{bmatrix},\
	M_{\Omega_{3}}^{-1}MP_Y \!=\!\begin{bmatrix}
	0.75\\0.25
	\end{bmatrix}.
	\end{align*}
	Since the first two vectors have zero element, $P_{XY}\notin \mathcal{H}_{XY}$.
\end{example}
\begin{remark}\label{ann}
	For $P_{XY}\in\mathcal{H}_{XY}$, $M_{\Omega}^{-1}MP_Y\in \mathbb{R}_{++}^{|\cal X|}$ implies $M_{\Omega}^{-1}MP_Y+\epsilon M_{\Omega}^{-1}M\begin{bmatrix}
	P_{X|Y_1}^{-1}J_u\\0\end{bmatrix}\in \mathbb{R}_{++}^{|\cal X|}$, furthermore, if the vector $M_{\Omega}^{-1}MP_Y$ contains a negative element, the vector $M_{\Omega}^{-1}MP_Y+\epsilon M_{\Omega}^{-1}M\begin{bmatrix}
	P_{X|Y_1}^{-1}J_u\\0\end{bmatrix}$ contains a negative element as well (i.e., not a feasible distribution), since we assumed $\epsilon$ is sufficiently small. 
\end{remark}
From Proposition~\ref{kos} and Remark~\ref{ann}, we conclude that each basic feasible solution of $\mathbb{S}_u$, i.e., $V_{\Omega}^*$, can be written as summation of a standard probability vector (built by $M_{\Omega}^{-1}MP_Y$) and a perturbation vector (built by $\epsilon M_{\Omega}^{-1}M\begin{bmatrix}
P_{X|Y_1}^{-1}J_u\\0\end{bmatrix}$). This is the key property to locally approximate $H(Y|U)$. 
\begin{remark}
	The number of basic feasible solutions of $\mathbb{S}_u$ is at most $\binom{|\mathcal{Y}|}{|\mathcal{X}|}$, thus we have at most $\binom{|\mathcal{Y}|}{|\mathcal{X}|}^{|\mathcal{Y}|}$ optimization problems with variables $P_U(.)$ and $J_u$.
\end{remark}
\subsection{Quadratic Optimization Problem}
In this part, we approximate the main problem in \eqref{equi} to form a new quadratic problem. Using \eqref{defin} and Proposition~\ref{kos} shows that the extreme points of $\mathbb{S}_u$ are perturbations of fixed distributions. Thus, the entropy of the extreme points, i.e., $H(V_{\Omega}^*)$, can be approximated using the local approximation of the entropy. 
In more detail, the approximation of $H(V_{\Omega}^*)$ is studied in Lemma~\ref{5} which is provided in Appendix~E. Furthermore, the geometrical properties of the extreme points are studied in Section~\ref{disc}. In next theorem we approximate \eqref{equi} using Lemma~\ref{5}. For simplicity for all $u\in\{1,..,|\cal U|\}$, we use $P_u$ instead of $P_U(u)$. 

\begin{theorem}\label{th11}
	Let $P_{XY}\in\mathcal{H}_{XY}$ and $V_{\Omega_u}^*\in\mathbb{S}_u^*,\ u\in\{1,..,|\mathcal{Y}|\}$. For sufficiently small $\epsilon$, the minimization problem in \eqref{equi} can be approximated as follows
	\begin{align}\label{minmin}
	&\min_{P_U(.),\{J_u, u\in\mathcal{U}\}} -\left(\sum_{u=1}^{|\mathcal{Y}|} P_ub_u+\epsilon P_ua_uJ_u\right)\\\nonumber
	&\text{subject to:}\\\nonumber
	&\sum_{u=1}^{|\mathcal{Y}|} P_uV_{\Omega_u}^*=P_Y,\ \sum_{u=1}^{|\mathcal{Y}|} P_uJ_u=0,\ P_U\in \mathbb{R}_{+}^{|\cal Y|},\\\nonumber
	&\sum_{i=1}^{|\mathcal{X}|} |J_u(i)|\leq 1,\  \sum_{i=1}^{|\mathcal{X}|}J_u(i)=0,\ \forall u\in\mathcal{U},
	\end{align} 
	where $a_u$ and $b_u$ are defined in Lemma~\ref{5}.
\end{theorem}
\begin{proof}
	The proof directly follows from Proposition~\ref{prop1111}, Proposition~\ref{4} and Lemma~\ref{5}. For $P_{Y|U=u}=V_{\Omega_u}^*,\ u\in\{1,..,|\mathcal{Y}|\}$, $H(Y|U)$ can be approximated as follows
	\begin{align*}
	H(Y|U) = \sum_u P_uH(P_{Y|U=u})\cong \sum_{u=1}^{|\mathcal{Y}|} P_ub_u+\epsilon P_ua_uJ_u.
	\end{align*} 
\end{proof}
\begin{remark}
 The weights $P_u$, which satisfy the constraints $\sum_{u=1}^{|\mathcal{Y}|} P_uV_{\Omega_u}^*=P_Y$ and $P_U(.)\geq 0$, form a standard probability distribution, since the sum of elements in each vector $V_{\Omega_u}^*\in\mathbb{S}_u^*$ equals to one.
\end{remark}
\begin{remark}
	If we set $\epsilon=0$, \eqref{minmin} becomes the same linear program as presented in \cite{borz}, since the term with $J_u$ disappears.
\end{remark}
\begin{proposition}
	The feasible set of \eqref{minmin} is non-empty. 
\end{proposition}
\begin{proof}
	Let $J_u=0$ for all $u\in{\mathcal{U}}$. In this case all sets $\mathbb{S}_u$ become the same sets named $\mathbb{S}$. Since the set $\mathbb{S}$ is an at most $|\mathcal{Y}|-1$ dimensional polytope and $P_Y\in\mathbb{S}$, $P_Y$ can be written as convex combination of some extreme points in $\mathbb{S}$. Thus, the constraint $\sum_{u=1}^{|\mathcal{Y}|} P_uV_{\Omega_u}^*=P_Y$ is feasible. Furthermore, by choosing $J_u=0$ for all $u\in{\mathcal{U}}$ other constraints in \eqref{minmin} are satisfied. 
\end{proof}
\subsection{Equivalent linear programming problem}\label{c}
In this part, we show that \eqref{minmin} can be rewritten as a linear program. To do so, consider the vector
$\eta_u=P_u\left(M_{\Omega_u}^{-1}MP_Y\right)+\epsilon \left(M_{\Omega_u}^{-1}M(1:|\mathcal{X}|)P_{X|Y_1}^{-1}\right)(P_uJ_u)$ for all $u\in \mathcal{U}$, where $\eta_u\in\mathbb{R}^{|\mathcal{X}|}$. The vector $\eta_u$ corresponds to multiple of non-zero elements of the extreme point $V_{\Omega_u}^*$. Furthermore, $P_u$ and $J_u$ can be uniquely found as
\begin{align*}
P_u&=\bm{1}^T\cdot \eta_u,\\
J_u&=\frac{P_{X|Y_1}M(1:|\mathcal{X}|)^{-1}M_{\Omega_u}[\eta_u-(\bm{1}^T \eta_u)M_{\Omega_u}^{-1}MP_Y]}{\epsilon(\bm{1}^T\cdot \eta_u)}.
\end{align*}
\begin{proposition}
\eqref{minmin} can be rewritten as a linear program using the vector $\eta_u$.	
\end{proposition}
\begin{proof}
	For the cost function we have
	\begin{align*}
	&-\left(\sum_{u=1}^{|\mathcal{Y}|} P_ub_u+\epsilon P_ua_uJ_u\right)=
	-\sum_u b_u(\bm{1}^T\eta_u)\\&-\epsilon\sum_u \!\!a_u\! \left[P_{X|Y_1}M(1:|\mathcal{X}|)^{-1}M_{\Omega_u}\![\eta_u\!-\!(\bm{1}^T \eta_u)M_{\Omega_u}^{-1}MP_Y]\right]\!,
	\end{align*}
	which is a linear function of elements of $\eta_u$ for all $u\in\mathcal{U}$.
	Non-zero elements of the vector $P_uV_{\Omega_u}^*$ equal to the elements of $\eta_u$, i.e., we have
	$
	P_uV_{\Omega_u}^*(\omega_i)=\eta_u(i).
	$
	Thus, the constraint $\sum_{u=1}^{|\mathcal{Y}|} P_uV_{\Omega_u}^*=P_Y$ can be rewritten as linear function of elements of $\eta_u$. For the constraints $\sum_{u=1}^{|\mathcal{Y}|} P_uJ_u=0$, $P_u\geq 0,\forall u$ and $\sum_{i=1}^{|\mathcal{X}|}J_u(i)=0$ we have
	\begin{align*}
	&\sum_u P_uJ_u=0\Rightarrow\\
	&\sum_uP_{X|Y_1}M(1:|\mathcal{X}|)^{-1}M_{\Omega_u}\left[\eta_u-(\bm{1}^T\cdot\eta_u)M_{\Omega_u}^{-1}MP_Y\right]=0,\\
	&P_u\geq 0,\ \forall u \Rightarrow \sum_i \eta_u(i)\geq 0,\ \forall u,\\
	&\bm{1}^TJ_u=0 \Rightarrow \\
	&\bm{1}^T P_{X|Y_1}M(1:|\mathcal{X}|)^{-1}M_{\Omega_u}\left[\eta_u-(\bm{1}^T\cdot\eta_u)M_{\Omega_u}^{-1}MP_Y\right]=0.
	\end{align*}
	Furthermore, for the last constraint we have 
	\begin{align*}
	&\sum_{i=1}^{|\mathcal{X}|} |J_u(i)|\leq 1\Rightarrow \\&\sum_i\! \left|\left(P_{X|Y_1}M(1\!:\!|\mathcal{X}|)^{-1}M_{\Omega_u}\!\left[\eta_u\!-\!(\bm{1}^T\!\eta_u)M_{\Omega_u}^{-1}MP_Y\right]\right)(i)\right|\! \leq\\&\epsilon(\bm{1}^T\eta_u),\ \forall u.
	\end{align*}
	The last constraint includes absolute values that can be handled by considering two cases for each absolute value. Thus, all constraints can be rewritten as linear function of elements of $\eta_u$ for all $u$.
\end{proof}
In the following, we provide an example where the procedure of finding the mechanism to produce $U$ is explained.

\begin{example}\label{ex2}
	Consider the leakage matrix $P_{X|Y}=\begin{bmatrix}
	0.3  \ 0.8 \ 0.5 \ 0.4\\0.7 \ 0.2 \ 0.5 \ 0.6
	\end{bmatrix}$  
	and $P_Y=[\frac{1}{2},\ \frac{1}{4},\ \frac{1}{8},\ \frac{1}{8}]^T$, furthermore, $\epsilon = 10^{-2}$.
	By using our method, the approximate solution to \eqref{minmin} is as follows
	\begin{align*}
	&P_U = \begin{bmatrix}
	0.7048 \\ 0.1492 \\ 0.146 \\ 0
	\end{bmatrix},\
	J_1 = \begin{bmatrix}
	-0.0023 \\ 0.0023
	\end{bmatrix},\ J_2 = \begin{bmatrix}
	0.5 \\ -0.5
	\end{bmatrix}\\
	&J_3 = \begin{bmatrix}
	-0.5 \\ 0.5
	\end{bmatrix},\
	J_4 = \begin{bmatrix}
	0 \\ 0
	\end{bmatrix},\ \min \text{(cost)} = 0.8239.
	\end{align*}
	Thus, we have $\max I(U;Y)\cong 0.9261$. For $\epsilon=0$, the maximum mutual information found in \cite{borz} is $0.9063$ which is less than our result. The details of the procedure on how to find the approximate solution is provided in Appendix~F.
\end{example}

\section{geometric interpretation and discussion}\label{disc}
Let $\mathbb{S} = \{y\in\mathbb{R}^{|\mathcal{Y}|}|My=MP_Y,\ y\geq 0\}$ and $\mathbb{S}^*$ be the set of the extreme points of $\mathbb{S}$. As shown in \cite{borz} $\mathbb{S}^*$ is found as follows:\\
	Let $\Omega = \{\omega_1,..,\omega_{|\mathcal{X}|}\}$, where $\omega_i\in\{1,..,|\mathcal{Y}|\}$ and all elements are arranged in an increasing order. If the vector $M_{\Omega}^{-1}MP_Y$ contains non-negative elements, then
\begin{align}\label{s^*}
V_{\Omega}^*(\omega_i)= (M_{\Omega}^{-1}MP_Y)(i).
\end{align}
and other elements of the extreme point $V_{\Omega}^*$ are zero. Note that if $P_{XY}\in\mathcal{H}_{XY}$ and by using Remark~3, if the set $\Omega$ produce an extreme point for $\mathbb{S}$, then it produces an extreme point for the set $\mathbb{S}_u$ as well. Now let $\mathbb{S}^*_u(i)\in\mathbb{S}_u^*$ and $\mathbb{S}^*(i)\in\mathbb{S}^*$. Also assume that they are both produced by the same set $\Omega$. By using \eqref{s_u^*} and \eqref{s^*} the relation between $\mathbb{S}^*(i)$ and $\mathbb{S}_u^*(i)$ is as follows
\begin{align}\label{ajab}
\mathbb{S}_u^*(i)= \mathbb{S}^*(i)+\epsilon V,
\end{align}
where
\begin{align}\label{v}
V(\omega_i)= M_{\Omega}^{-1}M\begin{bmatrix}
P_{X|Y_1}^{-1}J_u\\0\end{bmatrix}(i)
\end{align}
and other elements of $V$ are zero. Thus, the extreme points $\mathbb{S}_u^*$ are perturbed vectors of the extreme points  $\mathbb{S}^*$. 
\begin{remark}
	$\mathbb{S}^*$ are the extreme points of $\mathbb{S}_u$ when we have zero leakage, i.e., $\epsilon=0$. In this case, $\mathbb{S}_u=\mathbb{S}$ for all $u$.
\end{remark}
Next we show that the $\ell_1$ norm of $V$ is bounded from above.
\begin{proposition}\label{l1}
	Let $A=M_{\Omega}^{-1}M(1:|\mathcal{X}|)
	P_{X|Y_1}^{-1}$ and $a_i$ be the $i$-th column of $A$, i.e., $A=[a_1,..,a_{|\mathcal{X}|}]$. Then, if $J_u$ satisfies \eqref{prop3}, we have
	\begin{align*}
	\left\lVert V\right\rVert_1\leq r 
	\end{align*}
	where $r = \max_i 1^T|a_i|$ and $|a_i|=\begin{bmatrix} |a_i(1)|\\.\\.\\|a_i(|\mathcal{X}|)|\end{bmatrix}$. 
\end{proposition}
\begin{proof}
We have $\left\lVert V\right\rVert_1=\left\lVert AJ_u\right\rVert_1$, since other elements of $V$ are zero. 
\begin{align*}
\left\lVert V\right\rVert_1&\stackrel{(a)}{=}\left\lVert AJ_u\right\rVert_1=\sum_{i=1}^{|\mathcal{X}|} \left\lVert a_iJ_u(i)\right\rVert_1=\sum_{i=1}^{|\mathcal{X}|} |J_u(i)|\left\lVert a_i\right\rVert_1\\
&\stackrel{(b)}{\leq} \max_i |a_i| \sum_{i=1}^{|\mathcal{X}|} |J_u(i)|\stackrel{(c)}{\leq}\max_i |a_i|,
\end{align*}
where (a) follows from \eqref{v} and (b) comes from triangle inequality. Finally, (c) is due to \eqref{prop3}, i.e., privacy criterion. 
\end{proof}
\begin{remark}
	From \eqref{ajab} and Proposition~\ref{l1}, it can be seen that the extreme point $\mathbb{S}_u^*(i)$ is inside an $\ell_1$ ball of radius $r$ with center $\mathbb{S}^*(i)$.   
\end{remark}
Figure~\ref{pos} illustrates the possible positions of the extreme points $\mathbb{S}_u^*$. Since we assumed that $\epsilon$ is sufficiently small, if $\mathbb{S}^*$ is not inside the standard probability simplex, i.e. not feasible, $\mathbb{S}_u^*$ is not feasible as well.    
\begin{figure}[]
	\centering
	\includegraphics[width = 0.25\textwidth]{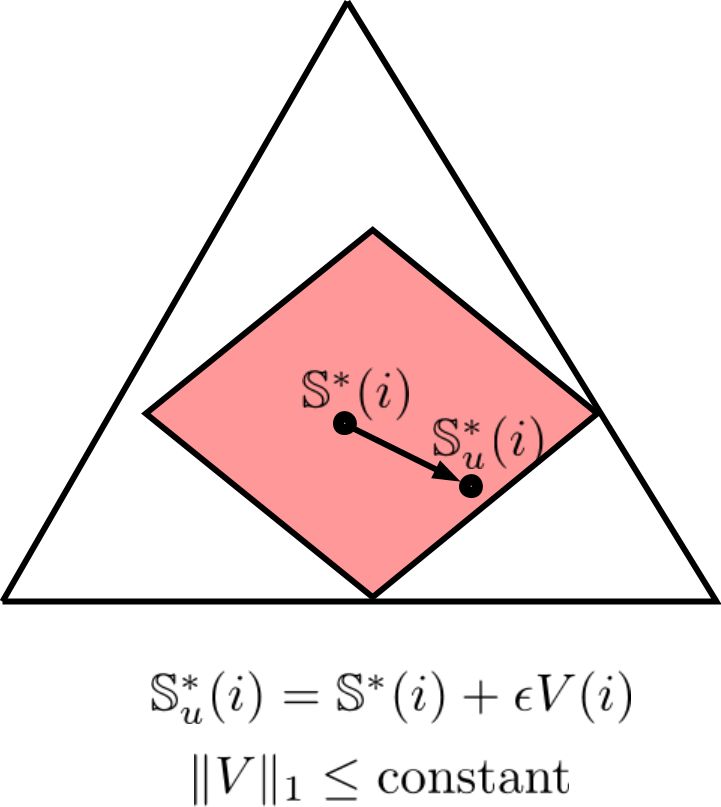}
	\caption{Possible positions of the extreme points. $\mathbb{S}_u^*$ is inside an $\ell_1$ ball of radius $r$ with center $\mathbb{S}^*(i)$.}
	\label{pos}
\end{figure}
Next, we investigate the range of permissible $\epsilon$. 
\begin{proposition}
	An upper bound for a valid $\epsilon$ in our work is as follows
	\begin{align*}
	\epsilon < \min\{\epsilon_1,\epsilon_2\},
	\end{align*} 
	where $\epsilon_1 = \frac{\min_{\Omega\in \Omega^2} \max_{y:M_{\Omega}^{-1}MP_Y(y)<0} |M_{\Omega}^{-1}MP_Y(y)|}{\max_{\Omega\in \Omega^2} |\sigma_{\max} (H_{\Omega})|}$, $\epsilon_2=\frac{\min_{y,\Omega\in \Omega^1} M_{\Omega}^{-1}MP_Y(y)}{\max_{\Omega\in \Omega^1} |\sigma_{\max} (H_{\Omega})|}$, $H_{\Omega}=M_{\Omega}^{-1}M(1:|\mathcal{X}|)P_{X|Y_1}^{-1}$ and $\sigma_{\max}$ is the largest right singular value.
\end{proposition}
\begin{proof}
	The proof is provided in Appendix~G.
\end{proof}
In the next section, we study an example motivated by simple watermarking. First we solve the problem using our approach and then we investigate how good our approximation performs. Later, we consider different metrics as utility and privacy leakage such as probability of error and MMSE to evaluate our method. We compare our results with perfect privacy approach in \cite{borz} and the upper bound obtained in \cite{MMSE}.
\section{Example}\label{example}
	Let $(Z_1,Z_2)$ be two images which are correlated and $X$ be a label added to the images. Let $(Y_1,Y_2)=(f(Z_1,X),f(Z_2,X))$ be two images produced by placing the label $X$ on $Z_1$ and $Z_2$, where the function $f(.)$ denotes putting the label on each original image. 
	We assume that the label $X$ contains sensitive information, thus, describes the private data and $(Y_1,Y_2)$ denote the useful data. Furthermore, we assume that the original images $(Z_1,Z_2)$ are not accessible and the goal is to find $U$ which reveals as much information as possible about $(Y_1,Y_2)$ and fulfills the privacy criterion. The utility is measured by $I(U;Y_1,Y_2)$ and the strong $\ell_1$-privacy criterion is employed to measure the privacy leakage. 
	Thus, the problem can be summarized as follows
	\begin{align*}
	\max_{P_{U|(Y_1,Y_2)}} \ \ &I(U;Y_1,Y_2),\\
	\text{subject to:}\ \ &X-(Y_1,Y_2)-U,\\
	& \left\lVert P_{X|U=u}-P_X \right\rVert_1\leq \epsilon,\ \forall u\in\mathcal{U}.\label{local}
	\end{align*}
	In order to study a numerical example consider the binary random variables $Z_1$, $Z_2$ and $X$ with probability distributions $p_X(x)=\begin{cases}
	\frac{3}{5},\ \ x=1\\\frac{2}{5},\ \ x=2
	\end{cases}$ and $p_{Z_i}(z)=\begin{cases}
	\frac{1}{3},\ \ z=1\\\frac{2}{3},\ \  z=2
	\end{cases}$ for $i=1,2$, have the patterns shown in Fig.~\ref{anna}, where each pixel is on or off. 
	\begin{figure}[H]
		\centering
		\includegraphics[scale=0.12]{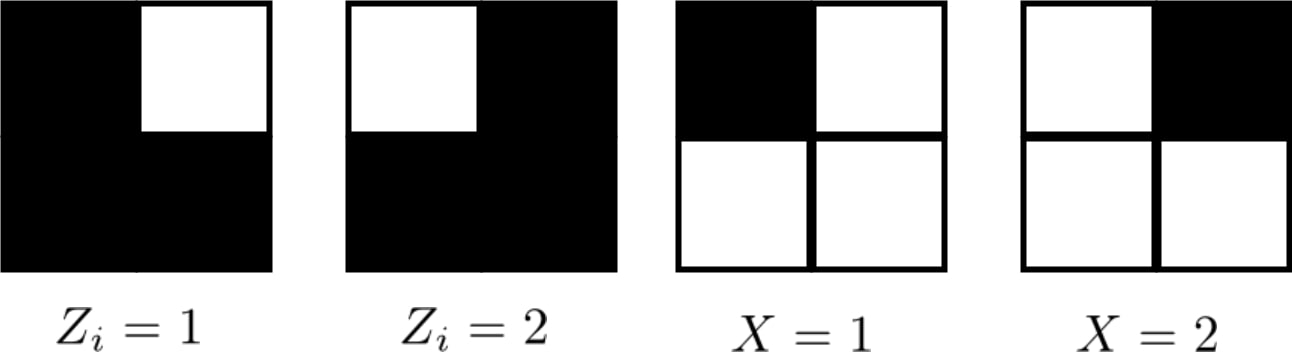}
		\caption{Possible patterns for $Z_1$, $Z_2$ and $X$.}
		\label{anna}
	\end{figure}   
	The correlation between $Z_1$ and $Z_2$ is given by the matrix $p_{Z_2|Z_1}=\begin{bmatrix}
	0.85 &0.3\\0.15 &0.7
	\end{bmatrix}$ and the function $f(.)$ is considered as XOR between the pixels. Thus, possible patterns for $(Y_1,Y_2)$ are as follows 
	\begin{figure}[H]
		\centering
		\includegraphics[scale = 0.12]{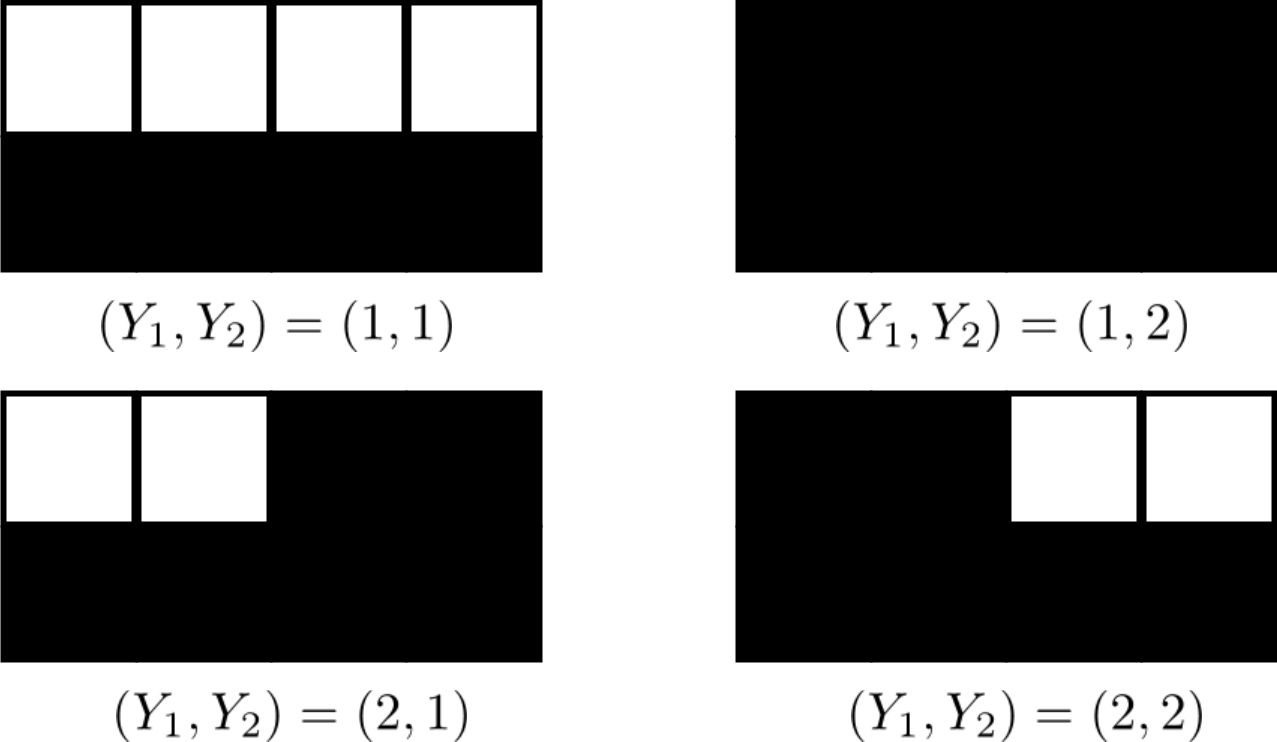}
		\label{f5}
		\caption{Possible patterns for the joint $(Y_1,Y_2)$.}
	\end{figure}
	with the marginal probability distribution 
	\begin{align*}
	p_{Y_1,Y_2}=\begin{cases}
	0.3567\ \ &(y_1,y_2)=(1,1)\\0.3933\ \ &(y_1,y_2)=(1,2)\\0.11\ \ &(y_1,y_2)=(2,1)\\0.14\ \ &(y_1,y_2)=(2,2)
	\end{cases}.
	\end{align*}
	Furthermore, the leakage matrix $P_{X|Y_1,Y_2}$ can be calculated as $P_{X|Y_1,Y_2}=\begin{bmatrix}
	0.4766 &0.7119 &0.2727 &0.8571\\0.5234 &0.2881 &0.7273 &0.1429
	\end{bmatrix}$ which is of full row rank. We use our method to find the approximate solution. 
	For $\epsilon=0.0562$, the approximate solution can be found and we obtain $P_U$ and $J_u$ as $P_U=[0,\ 0.4758,\ 0.486,\ 0.0382]^T$, $J_1=[0,\ 0]^T$, $J_2=[-0.5,\ 0.5]^T$, $J_3=[0.5,\ -0.5]^T$ and $J_4=[-0.1343, \ 0.1343]^T$. The minimum cost can be found as $0.5426$ and hence the maximum value of the main problem is approximately $H(Y)-0.4611=0.7102$. For $\epsilon=0$, we obtain the same linear program as in \cite{borz} and the maximum value is $0.6413$. 
	Next, we evaluate our approximation for different leakages.
	\subsection{How good is the approximation?}
	In this part, we sweep $\epsilon$ to compare our approximate solution with the exact solution found by exhaustive search and the perfect privacy solution. 
	In Fig.~\ref{fdid}, conditional entropy for different values of $\epsilon$ is illustrated. In this figure, the exact solution of the main problem and the approximate solution achieved by our approach is shown, where we used the exhaustive search for finding the exact minimum value.  We can see that the approximation error becomes small in the high privacy regime as expected. Furthermore, it can be seen that in the high privacy regime, the approximate solution tends to the perfect privacy solution.    
	\begin{figure}[H]
		\centering
			\includegraphics[scale=0.14]{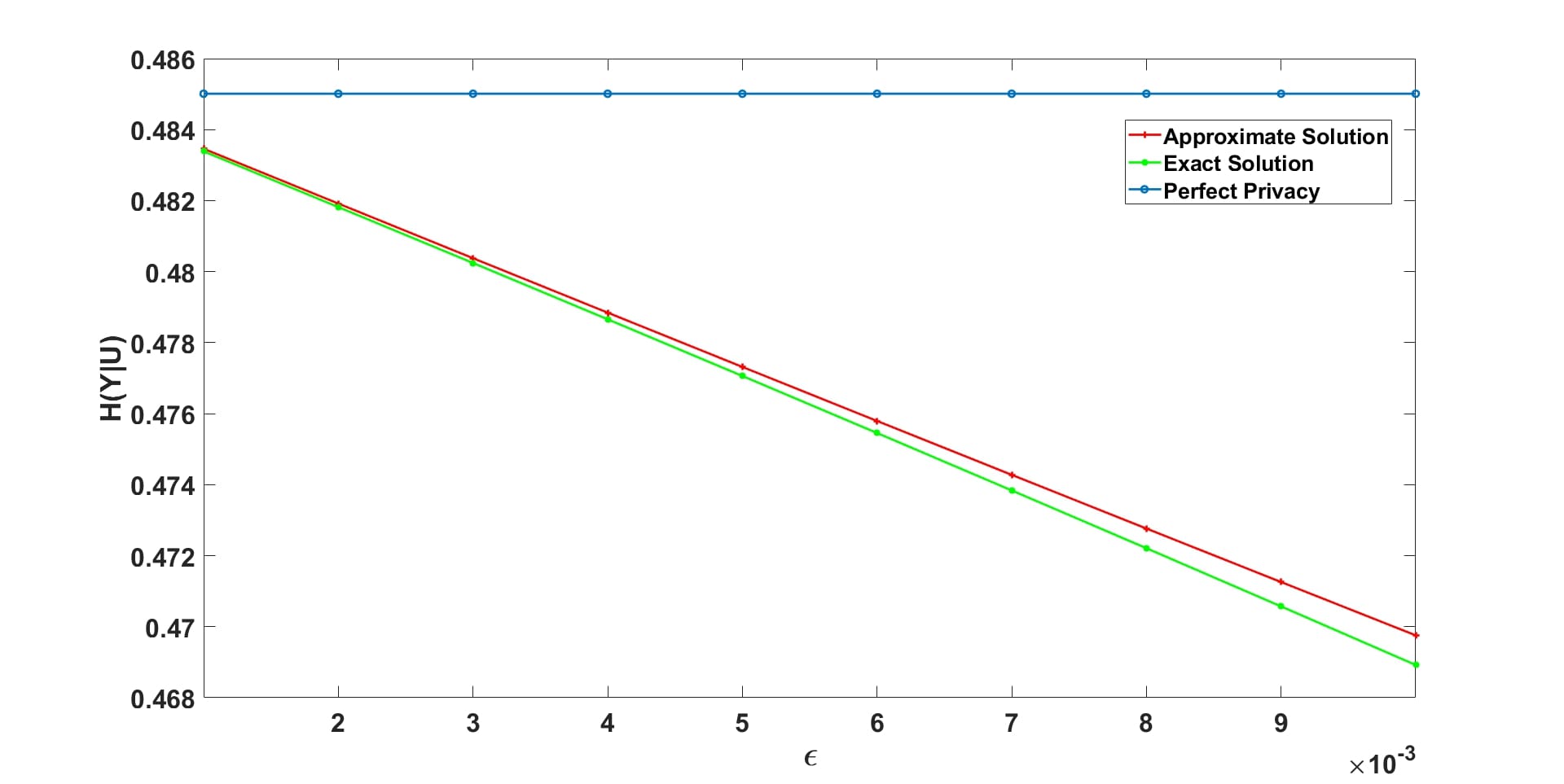}
		\caption{$H(Y|U)$ for different $\epsilon$, where we compare our proposed approximate solution with the exact solution found by an exhaustive search. Furthermore, in high privacy regime, our proposed solution tends to the solution for perfect privacy.}
		\label{fdid}
	\end{figure}
	Next, we investigate the role of correlation between $Z_1$ and $Z_2$ on the privacy-utility trade-off. Assume with probability $\alpha$, $Z_1=Z_2$ with mass function $p_{Z_i}(z)=\begin{cases}
	\frac{1}{3},\ \ z=1\\\frac{2}{3},\ \  z=2
	\end{cases}$ and with probability $1-\alpha$ we have $p_{Z_2|Z_1}=\begin{bmatrix}
	0.85 &0.3\\0.15 &0.7
	\end{bmatrix}$ and $p_{Z_1}(z)=\begin{cases}
	\frac{1}{3},\ \ z=1\\\frac{2}{3},\ \  z=2
	\end{cases}$. Here $\alpha$ corresponds to the correlation between $Z_1$ and $Z_2$ and hence $Y_1$ and $Y_2$. Thus, the kernel can be calculated as $p_{Z_2|Z_1}=\begin{bmatrix}
	0.85+0.15\alpha &0.3(1-\alpha)\\0.15(1-\alpha) &0.7+0.3\alpha 
	\end{bmatrix}$. Furthermore, we can calculate $P_{X|Y_1,Y_2}$ and $P_{Y_1Y_2}$ as
	\begin{align*}
	P_{Y_1Y_2}\!\!=\!\! \begin{bmatrix}
	\frac{11\alpha}{100}+\frac{107}{300}\\\frac{7\alpha}{50}+\frac{59}{150}\\\frac{11(1-\alpha)}{100}\\\frac{7(1-\alpha)}{50}
	\end{bmatrix}\!\!, 
	 P_{X|Y_1,Y_2}\!\!\!=\!\! \begin{bmatrix}
	\frac{9\alpha + 51}{33\alpha + 107} & \frac{18\alpha + 42}{21\alpha + 59}& \frac{3}{11} & \frac{6}{7}\\
	\frac{24\alpha + 56}{33\alpha + 107}& \frac{3\alpha + 17}{21\alpha + 59} &\frac{8}{11} &\frac{1}{7}
	\end{bmatrix}\!\!,
	\end{align*}
	Fig.~\ref{fdid1} illustrates the comparisons of $I(Y_1,Y_2;U)$ for different values of $\alpha$. We can see that the mutual information decreases with increasing $\alpha$. Intuitively, as $\alpha$ tends to $1$, $Y_1$ and $Y_2$ become more similar and the mutual information $I(U;Y_1,Y_2)$ tends to $I(U;Y_1)$. Furthermore, it can be seen that our approach achieves a better utility than perfect privacy for different values of $\alpha$. Furthermore, the utility achieved by exhaustive search is close to the utility attained by our approximation method.
	\begin{figure}[H]
		\centering
		\includegraphics[scale = 0.14]{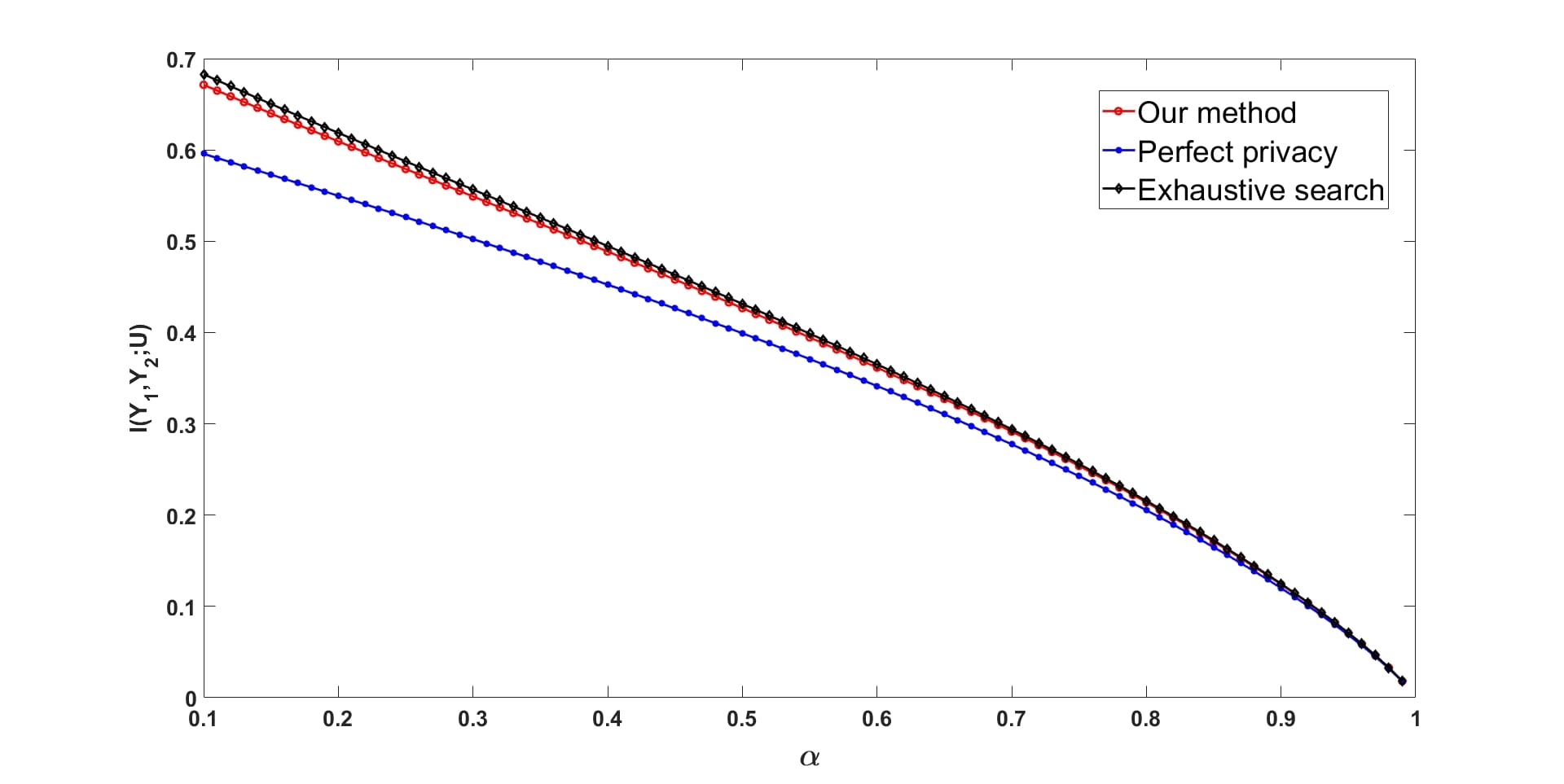}
		\caption{$I(Y|U)$ versus $\alpha$. As $\alpha$ increases, the mutual information (utility) decreases. As $\alpha$ increases $I(Y_1,Y_2;U)$ tends to $I(Y_1;U)$ and hence decreases.}
		\label{fdid1}
	\end{figure}
For the remaining parts we assume $Y=(Y_1,Y_2)$ with support $\{1,2,3,4\}$.
	\subsection{Privacy-utility trade-off for different measures}
	In this section, we consider two scenarios, where in the first scenario our approach is used to find the sub-optimal kernel $P_{U|(Y)}$ and in the second scenario the perfect privacy approach from \cite{borz} is used. Intuitively, $U$ and $Y$ should be as much dependent as possible under the privacy constraint. 
	First, we consider the probability of error using the MAP decision making rule as a measure of utility followed by the  normalized MMSE$(Y|U)$. Furthermore, in second part we compare our results with the upper bound found in \cite[Corrolary~2]{MMSE}.

	\emph{Probability of error based on MAP v.s. $\frac{\text{MMSE}(X|U)}{\text{Var(X)}}$}:\\
In this part, we use the average probability of error between disclosed data $U$ and
	desired data $Y$ using the MAP decision making rule based on observation $U$ for estimating $Y$ as utility, and the normalized MMSE of estimating $X$ based on observation $U$ as privacy measure, i.e., $\frac{\text{MMSE}(X|U)}{\text{Var}(X)}$. In Fig.~\ref{farhad2}, the parameter $\alpha$, which corresponds
	to the correlation between $Z_1$ and $Z_2$ , is swept to illustrate the
	privacy-utility trade-off.
	Fig.~\ref{farhad2} illustrates the privacy-utility trade-off using error probability and normalized MMSE. As it can be seen by letting small leakage we get better utility than perfect privacy approach. 
	\begin{figure}[h]
		\centering
		\includegraphics[scale=0.14]{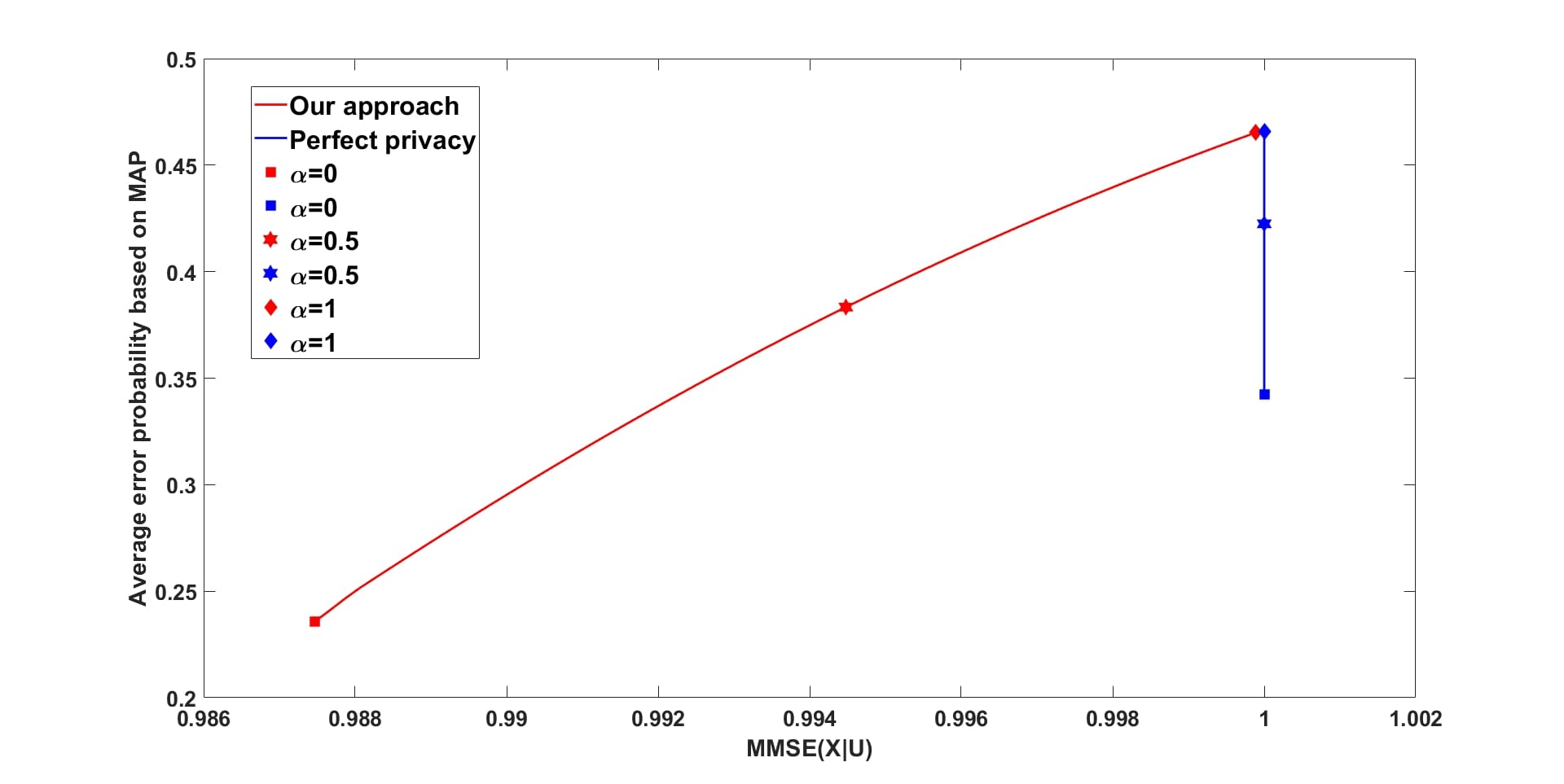}
		\caption{Privacy utility trade-off using probability of error as a measure of utility and normalized MMSE for measuring the privacy leakage. Red curve illustrates the trade-off using our approach and the blue curve shows the trade-off for perfect privacy.}
		\label{farhad2}
	\end{figure}

	\emph{$\frac{\text{MMSE}(Y|U)}{\text{Var}(Y)}$ v.s. $\frac{\text{MMSE}(X|U)}{\text{Var}(X)}$}:\\
	In this section, we use the normalized MMSE$(Y|U)$ as utility measure and normalized MMSE$(X|U)$ as privacy leakage measure, where
	\begin{align*}
	\text{MMSE}(Y|U)
	&\!=\!\! \sum_u P_U(u)\! \left(\mathbb{E}(Y^2|U=u)\!-\!(\mathbb{E}(Y|U=u))^2\right)\!,\\
	\text{MMSE}(X|U)
	&\!=\!\! \sum_u P_U(u)\! \left(\mathbb{E}(X^2|U=u)\!-\!(\mathbb{E}(X|U=u))^2\right)\!.
	\end{align*}
	We consider three cases where in first case we calculate $P_{U|Y}$ using our method, in second case we consider the upper bound of $\text{wESNR}_{\epsilon}$ studied in \cite[Corrolary~2]{MMSE}, and for the third case the perfect privacy approach is considered. We have
	\begin{align*}
	\text{wESNR}_{\epsilon}(X;Y)= \inf_{P_{U|Y}} &\frac{\text{MMSE}(Y|U)}{\text{var}(Y)}\\
	s.t.\ &X-Y-U,\\
	&\text{MMSE}(X|U)\geq (1-\epsilon)\text{var}(X),  
	\end{align*} 
	and from \cite[Corrolary~2]{MMSE},
	\begin{equation}
	\text{wESNR}_{\epsilon}(X;Y)\leq (1-\frac{1}{\eta_Y^2(X)})\min(\epsilon,\eta_Y^2(X)), \label{uppy}
	\end{equation}
	which is achieved by the following erasure channel
	\begin{align*}
	P_{U|Y}(u|y)=\begin{cases}
	1-\delta\ &u=y \\\delta\ &u=e,
	\end{cases}
	\end{align*}
	where $\delta=1-\frac{\epsilon}{\eta_Y^2(X)}$. We sweep $\alpha$ to illustrate the privacy-utility trade-off. In order to compare \eqref{uppy} with our results, we use $\epsilon^2$ in the upper bound.
	Fig~\ref{compare2} illustrates the privacy-utility trade-off of our method compared to the case of perfect privacy. Furthermore, the lowest upper bound in \eqref{uppy}, i.e., the best upper bound, is shown in this figure. 
	We can see that our method achieves better utility compared to the perfect privacy when a small leakage is allowed. Furthermore, the achieved utility is lower than the upper bound for all $\alpha\in[0,1]$.
	Next proposition gives a lower bound on MMSE$(X|U)$ when we use the strong $\ell_1$-privacy criterion.
	\begin{proposition}
		Let $X\in\{x_1,x_2\}$ be a binary r.v. with zero mean and $U$ satisfy the strong $\ell_1$-privacy criterion. Then we have 
		\begin{align*}
		\text{MMSE}(X|U)\geq \text{Var}(X)-\frac{1}{4}\epsilon^2(x_1-x_2)^2.
		\end{align*}
		Furthermore,
		\begin{align*}
		\text{Var}(X)\leq \frac{1}{4}(x_1-x_2)^2,
		\end{align*}
		and if $x_1=-x_2$
		\begin{align*}
		\text{Var}(X)= \frac{1}{4}(x_1-x_2)^2,
		\end{align*}
	\end{proposition}
	\begin{proof}
		The proof is provided in Appendix~I.
	\end{proof}
	\begin{remark}
		In other words, if $x_1=-x_2$, the privacy criterion used in this work results in the privacy criterion employed in \cite{MMSE}.
	\end{remark}
	\begin{figure}[H]
		\centering
		\includegraphics[scale=0.14]{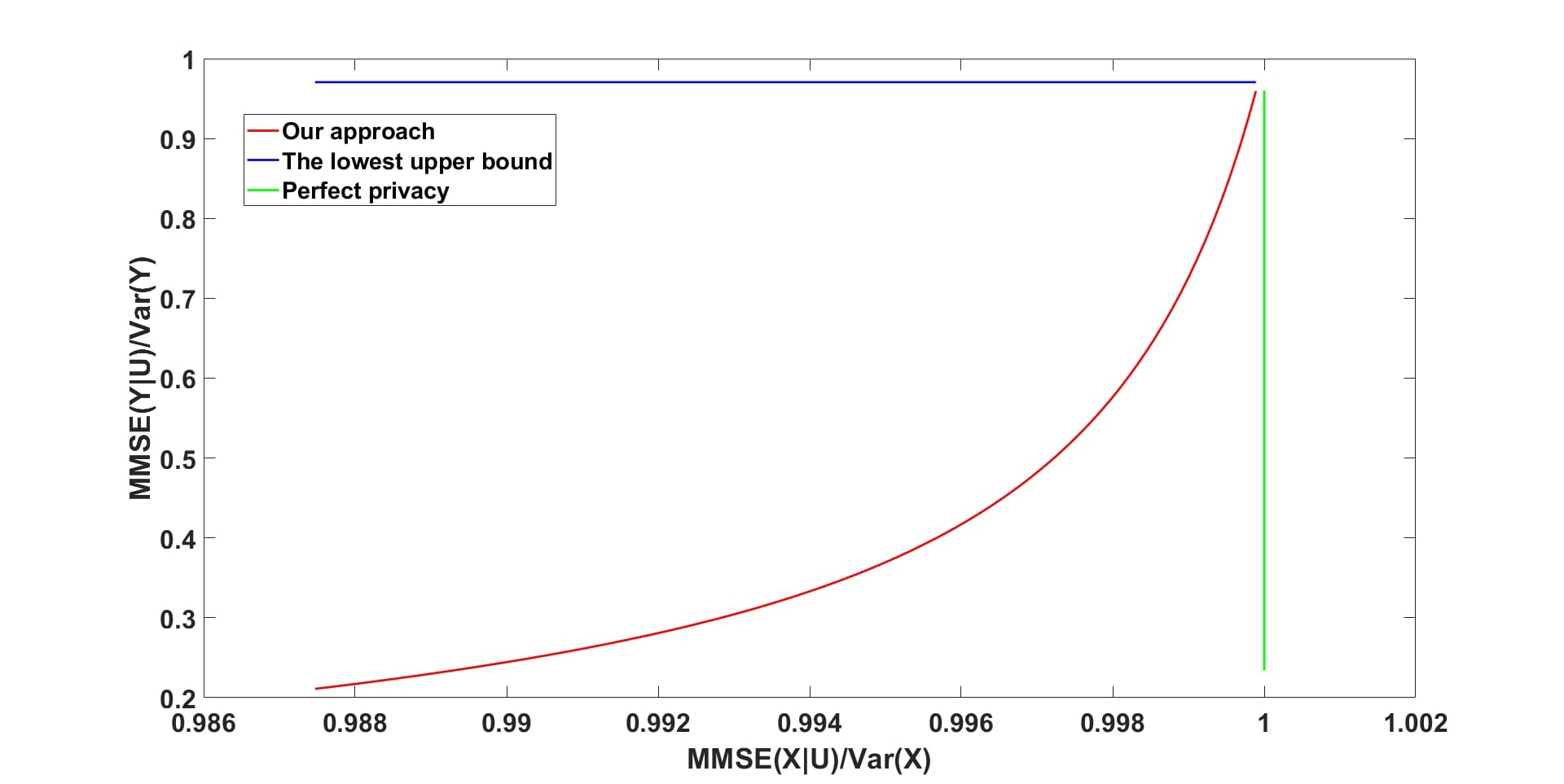}
		\caption{Privacy utility trade-off.}
		\label{compare2}
	\end{figure}
	In next part, we study the case $\alpha=1$ using the approach in \cite{khodam}.\\
	\emph{Special case for $\alpha=1$ (invertible leakage matrix)}:\\
In this part, we solve the main problem for $\alpha=1$. We have:
	\begin{align*}
	P_{Y_1Y_2}= \begin{bmatrix}
	\frac{7}{15}\\\frac{8}{15}
	\end{bmatrix},\ P_{X|Y_1,Y_2}= \begin{bmatrix}
	\frac{3}{7} & \frac{3}{4}\\
	\frac{4}{7}& \frac{1}{4} 
	\end{bmatrix},
	\end{align*}
	and the main problem is reduced to:
	\begin{align*}
	\max_{P_{U|Y}} \ \ &I(U;Y),\\
	\text{subject to:}\ \ &X-Y-U,\\
	& \left\lVert P_{X|U=u}-P_X \right\rVert_1\leq \epsilon,\ \forall u\in\mathcal{U}.
	\end{align*}
	Note that the leakage matrix is invertible and we can not use the method of this work to solve the problem, however by using our previous result in \cite{khodam} (since the leakage matrix is invertible) we can write:
	\begin{align*}
	I(U;Y)\cong\frac{1}{2}\epsilon^2\sum_u P_U\|[\sqrt{P_Y}^{-1}]P_{X|Y}^{-1}[\sqrt{P_X}]L_u\|^2,
	\end{align*}
	where $L_u=[\sqrt{P_X}^{-1}]J_u$. The privacy constraint can be written as
	$
	\|[\sqrt{P_X}]L_u\|_1\leq1.
	$
	The only difference with the problem in \cite{khodam} is the privacy constraint. The problem can be approximated by:
	\begin{align}
	\max_{\{L_u,P_U\}} \ &\sum_u P_U(u)\|W\cdot L_u\|^2,\\
	\text{subject to:}\ &\|L_u\|_1\leq 1,\ \forall u\in\mathcal{U},\\
	&\sum_x \sqrt{P_X(x)}L_u(x)=0,\ \forall u,\\
	&\sum_u P_U(u)\sqrt{P_X(x)}L_u(x)=0,\ \forall x,
	\end{align}
	where $W=[\sqrt{P_Y}]^{-1}P_{X|Y_1}^{-1}[\sqrt{P_X}]$. The maximizer of the final problem is $L_u = \frac{L^*}{\|[\sqrt{P_X}]L^*\|_1}$ and the maximum value is $\frac{\frac{1}{2}\epsilon^2\sigma_{\max}^2}{\|[\sqrt{P_X}]L^*\|_1^2}$, where $\sigma_{\max}$ is the maximum singular value of $W$ and $L^*$ is the corresponding singular vector. 
\section{conclusion}\label{concul1}
It has been shown that a statistical data disclosure problem can be decomposed if the leakage matrix is not of full rank. Furthermore, information geometry can be used to approximate $H(Y|U)$, which allows us to simplify the optimization problem for sufficiently small leakage. In particular, the new optimization problem can be rewritten as a linear program. In an example, we evaluated our method by comparing it with the exact solution and perfect privacy solution for different leakages. Furthermore, the privacy-utility trade-off has been studied for different measures of utility and privacy leakage. It is shown that by letting small leakage, our method achieves better utility compared to the perfect privacy approach.   
	\section*{Appendix A}\label{appa}
	As shown in Lemma~\eqref{null2}, $P_{Y|U=u}$ must belong to the following set for every $u\in\mathcal{U}$
	\begin{align*}
	\psi\!=\!\left\{\!y\in\mathbb{R}^{|\mathcal{Y}|}|My\!=\!MP_Y\!+\!\epsilon M\!\!\begin{bmatrix}
	P_{X|Y_1}^{-1}J_u\\0
	\end{bmatrix}\!,y\geq 0,J_u\in\mathcal{J} \!\right\}\!,
	\end{align*}
	where $\mathcal{J}=\{J\in\mathbb{R}^{|\mathcal{X}|}|\left\lVert J\right\rVert_1\leq 1,\ \bm{1}^{T}\cdot J=0\}$. First we show that $\psi$ is a compact set. 
	Using Lemma~\eqref{44}, each vector inside $\psi$ is a standard probability vector and thus the set $\psi$ is bounded.
	$\mathcal{J}$ corresponds to a closed set since each conditions $\left\lVert J\right\rVert_1\leq 1$ and $\bm{1}^{T}\cdot J=0$ represent a closed set and thus their intersection is closed. Let $\{y_1,y_2,..\}$ be a convergent sequence with $\lim_{i\rightarrow \infty} y_i=y^{*}$, which each $y_i$ is inside $\psi$. Since $\{y_i\}$ is convergent it is a Cauchy sequence and thus there exist $N$ such that for all $i,j\geq N$ we have $\left\lVert y_i-y_j\right\rVert\leq \epsilon_1$. Thus,
	$\left\lVert My_i-My_j\right\rVert\leq \epsilon'$ where $\epsilon'=\left\lVert M\right\rVert\epsilon_1$. Since $y_i,y_j\in \psi$ we obtain
	\begin{align*}
	&\left\lVert\epsilon M(1:|\mathcal{X}|)P_{X|Y_1}^{-1}J_{u_i}\!\!-\!\epsilon M(1:|\mathcal{X}|)P_{X|Y_1}^{-1}J_{u_j}\right\rVert\leq \epsilon',\\
	&\Rightarrow\epsilon\left\lVert \left(M(1:|\mathcal{X}|)P_{X|Y_1}^{-1}\right)^{-1}\right\rVert^{-1}\left\lVert J_{u_i}-J_{u_j}\right\rVert\leq \epsilon',\\
	&\Rightarrow\left\lVert J_{u_i}-J_{u_j}\right\rVert\leq \epsilon'',
	\end{align*}
	where $\epsilon''=\frac{\epsilon'}{\epsilon\left\lVert\left(M(1:|\mathcal{X}|)P_{X|Y_1}^{-1}\right)^{-1}\right\rVert^{-1}}$. In last line we used the fact that $M(1:|\mathcal{X}|)$ is invertible. Thus, $\{J_{u_i}\}$ is a Cauchy sequence. Since the set $\mathcal{J}$ is bounded and closed it is compact and thus the Cauchy sequence $\{J_{u_i}\}$ is convergent. Let $J^*=\lim_{i\rightarrow \infty} J_{u_i}$, for sufficiently large $i$ we have
	\begin{align*}
	\left\lVert J^*-J_{u_i}\right\rVert \leq \epsilon_2,\\\Rightarrow
	\left\lVert\epsilon M(1:|\mathcal{X}|)P_{X|Y_1}^{-1}J^*\!\!-\!\epsilon M(1:|\mathcal{X}|)P_{X|Y_1}^{-1}J_{u_i} \right\rVert\leq \epsilon_2'\\\Rightarrow
	\left\lVert MP_Y+\epsilon M(1:|\mathcal{X}|)P_{X|Y_1}^{-1}J^*-My_i\right\rVert\leq \epsilon_2'.
	\end{align*}
	We can conclude that $My^*$ is of the form $MP_Y+\epsilon M(1:|\mathcal{X}|)P_{X|Y_1}^{-1}J^*$ and thus $\psi$ is closed and compact. We define a vector mapping $\theta:\Psi\rightarrow \mathbb{R}^{|\cal Y|}$ as follows
	\begin{align*}
	\theta_i\left(p_{Y|U(\cdot|U)}\right)&=p_{Y|U}(y_i|u),\ i\in [1:|\mathcal{Y}|-1],\\
	\theta_{|\mathcal{Y}|}&=H(Y|U=u).
	\end{align*}
	Since the mapping $\theta$ is continuous and the set $\Psi$ is compact, by using Fenchel-Eggleston-Carath\'{e}odory's Theorem \cite{el2011network} for every $U$ with p.m.f $F(u)$ there exists a random variable $U'$ with p.m.f $ F(u')$ such that $|\cal U'|\leq |\cal Y|$ and collection of conditional p.m.fs $P_{Y|U'}(\cdot|u')\in \Psi$ where
	\begin{align*}
	\int_u \theta_i(p(y|u))dF(u)=\sum_{u'\in\cal U'}\theta_i(p(y|u'))p(u').
	\end{align*}
	It ensures that by replacing $U$ by $U'$, $I(U;Y)$ and the distribution $P_Y$ are preserved. Furthermore, the condition $\sum_{u'}P_{U'}(u')J_{u'}=\bm{0}$ is satisfied since $P_Y\in \psi$ and we have
	\begin{align*}
	P_Y=\sum_{u'} P_{U'}P_{Y|U'=u'}\Rightarrow P_X=\sum_{u'} P_{U'}P_{X|U'=u'},\\
	\sum_{u'} P_{U'}(P_{X|U'=u'}-P_X)=\bm{0}\Rightarrow \sum_{u'}P_{U'}(u')J_{u'}=\bm{0}.
	\end{align*}
	Note that any point in $\Psi$ satisfies the privacy criterion, i.e., the equivalent $U'$ satisfies the per-letter privacy criterion as well. Thus, without loss of optimality we can assume $|\mathcal{U}|\leq |\mathcal{Y}|$.
	
	Let $\mathcal{A}=\{P_{U|Y}(\cdot|\cdot)|U\in\mathcal{U}, Y\in\mathcal{Y},||\cal U|\leq |\cal Y| \}$ and $\mathcal{ A}_y=\{P_{U|Y}(\cdot|y)|U\in\mathcal{U},|\cal U|\leq |\cal Y| \}$, $\forall y\in \cal Y$. $\mathcal{A}_y$ is a standard $|\mathcal {U}|-1$ simplex and since $|\cal U|\leq |\cal Y|<\infty$ it is compact. Thus, $\mathcal{A}=\cup_{y\in\mathcal{Y}}\mathcal{A}_y$ is compact. And the set $\mathcal{A}'=\{P_{U|Y}(\cdot|\cdot)\in\mathcal{A}|X-Y-U, \left\lVert P_{X|U=u}-P_X\right\rVert_1\leq \epsilon,\ \forall u\}$ is a closed subset of $\mathcal{A}$ since $||.||_1$ is closed on the interval $[0,\epsilon]$. Therefore, $\mathcal {A}'$ is compact. Since $I(U;Y)$ is a continuous mapping over $\mathcal{A}'$, the supremum is achieved. Thus, we can use the maximum instead of the supremum.
	\section*{Appendix B}
	In the following lemma two properties of  null space Null$(P_{X|Y})$ are investigated.
	\begin{lemma}\label{null}
		Let $\beta$ be a vector in $\mathbb{R}^{|\mathcal{X}|}$. $\beta\in\text{Null}(P_{X|Y})$ if and only if $\beta\in\text{Null}(M)$, where $M\in \mathbb{R}^{|\mathcal{X}|\times|\mathcal{Y}|}$ is constructed as follows:
		Let $V$ be the matrix of right eigenvectors of $P_{X|Y}$, i.e., $P_{X|Y}=U\Sigma V^T$ and $V=[v_1,\ v_2,\ ... ,\ v_{|\mathcal{Y}|}]$, then $M$ is defined as
		\begin{align*}
		M \triangleq \left[v_1,\ v_2,\ ... ,\ v_{|\mathcal{X}|}\right]^T.  
		\end{align*}  
		Furthermore, if $\beta\in\text{Null}(P_{X|Y})$, then $1^T\beta=0$.
	\end{lemma}
	\begin{proof}
		Since the rank of $P_{X|Y}$ is $|\mathcal{X}|$, every vector $\beta$ in Null$(P_{X|Y})$ can be written as a linear combination of $\{v_{|\mathcal{X}|+1}$, ... ,$v_{|\mathcal{Y}|}\}$ and since each vector in $\{v_{|\mathcal{X}|+1}$, ... ,$v_{|\mathcal{Y}|}\}$ is orthogonal to the rows of $M$ we conclude that $\beta\in\text{Null}(M)$. If $\beta\in\text{Null}(M)$, $\beta$ is orthogonal to the vectors $\{v_1$,...,$v_{|\mathcal{X}|}\}$ and thus, $\beta\in \text{linspan}\{v_{|\mathcal{X}|+1},...,v_{|\mathcal{Y}|}\}=\text{Null}(P_{X|Y})$. Furthermore, for every $\beta\in\text{Null}(P_{X|Y})$, we can write
		$
		1^T\beta=1^TP_{X|Y}\beta=0.
		$
	\end{proof}
	The next lemma shows that if the Markov chain $X-Y-U$ holds and $J_u$ satisfies the three properties \eqref{prop1}, \eqref{prop2} and \eqref{prop3}, then $P_{Y|U=u}$ lies in a convex polytope.
	\begin{lemma}\label{null2}
		For sufficiently small $\epsilon>0$, for every $u\in\mathcal{U}$, the vector $P_{Y|U=u}$ belongs to the following convex polytope $\mathbb{S}_{u}$
		\begin{align*}
		\mathbb{S}_{u} = \left\{y\in\mathbb{R}^{|\mathcal{Y}|}|My=MP_Y+\epsilon M\begin{bmatrix}
		P_{X|Y_1}^{-1}J_u\\0
		\end{bmatrix},\ y\geq 0\right\},
		\end{align*}
		where $\begin{bmatrix}
		P_{X|Y_1}^{-1}J_u\\0
		\end{bmatrix}\in\mathbb{R}^{|\cal Y|}$ and $J_u$ satisfies \eqref{prop1}, \eqref{prop2}, and \eqref{prop3}.
	\end{lemma}
	\begin{proof}
		By using the Markov chain $X-Y-U$ we have
		\begin{align}
		P_{X|U=u}-P_X=P_{X|Y}[P_{Y|U=u}-P_Y]=\epsilon J_u.\label{jigar}
		\end{align}
		Let $\alpha=P_{Y|U=u}-P_Y$. By partitioning $\alpha$ into two parts $[\alpha_1\ \alpha_2]^T$ with sizes $|\cal X|$ and $|\cal Y|-|\cal X|$, respectively, from \eqref{jigar} we obtain
		\begin{align*}
		P_{X|Y_1}\alpha_1+P_{X|Y_2}\alpha_2=\epsilon J_u. 
		\end{align*}
		Using invertibilty of  $P_{X|Y_1}$, this implies that 
		\begin{align*}
		\alpha_1=\epsilon P_{X|Y_1}^{-1}J_u-P_{X|Y_1}^{-1}P_{X|Y_2}\alpha_2.
		\end{align*}
		Thus, we have
		\begin{align*}
		P_{Y|U=u} = P_Y+\epsilon\begin{bmatrix}
		P_{X|Y_1}^{-1}J_u\\0
		\end{bmatrix}
		+\begin{bmatrix}
		-P_{X|Y_1}^{-1}P_{X|Y_2}\alpha_2\\\alpha_2
		\end{bmatrix}.
		\end{align*}
		Note that the vector $\begin{bmatrix}
		-P_{X|Y_1}^{-1}P_{X|Y_2}\alpha_2\\\alpha_2
		\end{bmatrix}$ belongs to Null$(P_{X|Y})$ since we have
		\begin{align*}
		P_{X|Y}\begin{bmatrix}
		-P_{X|Y_1}^{-1}P_{X|Y_2}\alpha_2\\\alpha_2
		\end{bmatrix}&=[P_{X|Y_1}\ P_{X|Y_2}]\begin{bmatrix}
		-P_{X|Y_1}^{-1}P_{X|Y_2}\alpha_2\\\alpha_2
		\end{bmatrix}\\&=-P_{X|Y_1}P_{X|Y_1}^{-1}P_{X|Y_2}\alpha_2\!+\!P_{X|Y_2}\alpha_2\\&=\bm{0}.
		\end{align*}
		Thus, by Lemma~\ref{null} it belongs to Null$(M)$ so that we have
		\begin{align*}
		MP_{Y|U=u} = MP_Y+\epsilon M\begin{bmatrix}
		P_{X|Y_1}^{-1}J_u\\0
		\end{bmatrix}
		\end{align*}
		Consequently, if the Markov chain $X-Y-U$ holds and the perturbation vector $J_u$ satisfies \eqref{prop1}, \eqref{prop2}, and \eqref{prop3}, then we have $P_{Y|U=u}\in\mathbb{S}_u$.
	\end{proof}
	\begin{lemma}\label{44}
		Any vector $\alpha$ in $\mathbb{S}_u$ is a standard probability vector. Also, for any pair $(U,Y)$, for which $P_{Y|U=u}\in\mathbb{S}_u$, $\forall u\in\mathcal{U}$ with $J_u$ satisfying \eqref{prop1}, \eqref{prop2}, and \eqref{prop3}, we can have $X-Y-U$ and $P_{X|U=u}-P_X=\epsilon\cdot J_u$.  
	\end{lemma}
	\begin{proof}
		For the first claim it is sufficient to show that for any $\gamma\in\mathbb{S}_u$, we have $1^T\gamma=1$. Since $\gamma\in\mathbb{S}_u$, we have $M(\gamma-P_Y-\epsilon\begin{bmatrix}
		P_{X|Y_1}^{-1}J_u\\0
		\end{bmatrix})=0$ and from Lemma~\ref{null} we obtain $1^T(\gamma-P_Y-\epsilon\begin{bmatrix}
		P_{X|Y_1}^{-1}J_u\\0
		\end{bmatrix})=0$ which yields $1^T\gamma=1$. In the last conclusion we used the fact that $1^T\begin{bmatrix}
		P_{X|Y_1}^{-1}J_u\\0
		\end{bmatrix}=0$. This is true since by using \eqref{prop1} we have $1_{|\mathcal{Y}|}^T\begin{bmatrix}
		P_{X|Y_1}^{-1}J_u\\ 0 \end{bmatrix}=1_{|\mathcal{X}|}^T
		P_{X|Y_1}^{-1}J_u\stackrel{(a)}{=}1^TJ_u=0$, where $(a)$ follows from $1_{|\mathcal{X}|}^T
		P_{X|Y_1}^{-1}=1_{|\mathcal{X}|}^T\Leftrightarrow 1_{|\mathcal{X}|}^T
		P_{X|Y_1}=1_{|\mathcal{X}|}^T$ and noting that columns of $P_{X|Y_1}$ are vector distributions.  
		If $P_{Y|U=u}\in\mathbb{S}_u$ for all $u\in\mathcal{U}$, then we have $X-Y-U$ and 
		\begin{align*}
		&M(P_{Y|U=u}-P_Y-\epsilon\begin{bmatrix}
		P_{X|Y_1}^{-1}J_u\\0
		\end{bmatrix})=0, \\ \Rightarrow\ &P_{X|Y}(P_{Y|U=u}-P_Y-\epsilon\begin{bmatrix}
		P_{X|Y_1}^{-1}J_u\\0
		\end{bmatrix})=0,\\
		\Rightarrow\ &P_{X|U=u}-P_X=\epsilon\cdot P_{X|Y}(\begin{bmatrix}
		P_{X|Y_1}^{-1}J_u\\0
		\end{bmatrix})=\epsilon\cdot J_u,
		\end{align*}
		where in last line we use the Markov chain $X-Y-U$. 
		Furthermore, $J_u$ satisfies \eqref{prop1}, \eqref{prop2} and \eqref{prop3} for all $u\in\mathcal{U}$. Thus, the privacy criterion defined in \eqref{local} holds.
	\end{proof}
\section*{Appendix C}
As mentioned before it is sufficient to consider $U$ such that $|\mathcal{U}|\leq|\mathcal{Y}|$. Let $\mathbb{S}^*_u$ be the set of extreme points of $\mathbb{S}_u$. Then, we show that $P^*_{Y|U=u}\in\mathbb{S}^*_u$ for all $u\in\mathcal{U}$. Assume that the minimum of $H(Y|U)$ is achieved by points $P^*_1,...,P^*_N$, where $N\leq|\mathcal{Y}|$ and $P^*_i\in \mathbb{S}_i$ for all $1\leq i\leq N$. Furthermore, we have
$\sum_i P_U(i)P^*_i=P_Y$. Now consider $P^*_i$, which belongs to $\mathbb{S}_i$ and suppose that $P^*_i\notin \mathbb{S}^*_i$ (This is a proof by contradiction). Since the set $\mathbb{S}_i$ is a convex polytope with dimension at most $|\mathcal{Y}|-1$, we can write $P^*_i$ as a convex combination of at most $|\mathcal{Y}|$ points in $\mathbb{S}^*_i$, i.e., we have
$P^*_i = \sum_{l=1}^{|\mathcal{Y}|} \lambda_l \psi_l^*$,   
with
$
\psi_l^*\in\mathbb{S}^*_i,\sum_{l=1}^{|\mathcal{Y}|}\lambda_l=1,\ \lambda_l\geq 0,\ \forall l\in\{1,..,|\mathcal{Y}|\}.
$
Since we assume that $P^*_i\notin \mathbb{S}^*_i$, at least two of $\lambda_l$'s need to be non-zero. Furthermore, by using Jensen inequality we have
$
H(P_i^*)>\sum_{l=1}^{|\mathcal{Y}|}\lambda_lH(\psi_l^*),
$
where the inequality is strict since we assumed $P^*_i\notin \mathbb{S}^*_i$. We claim that the points $\{P_1^*,...,P_{i-1}^*,\psi_1^*,...,\psi_{|\mathcal{Y}|},P_{i+1}^*,...,P_N^*\}$ with weights $\{P_U(1),...,P_U(i-1),\lambda_1P_U(i),...,\lambda_{|\mathcal{Y}|}P_U(i),P_U(i+1),...,P_U(N)\}$ achieve lower entropy compared to $P^*_1,...,P^*_N$. $\{P_1^*,..,P_{i-1}^*,\psi_1^*,..,\psi_{|\mathcal{Y}|},P_{i+1}^*,..,P_N^*\}$ is feasible since we have
\begin{align*}
\sum_{i=1}^{N} \!P_U(i)P^*_i\!=\!P_Y\! \Rightarrow \!\!\!\!\!\sum_{\begin{array}{c}\substack{l=1,\\ l\neq i}\end{array}}^{N} \!\!\!P_U(l)P^*_{l}\!+\!\!\sum_{j=1}^{|\mathcal{Y}|} P_U(i)\lambda(j)\psi^*(j)\!=\!P_Y,
\end{align*}  
furthermore,
\begin{align*}
\sum_{j=1}^{N}\!\! P_U(j)H(P_j^*)\!\!>\!\!\!\!\!\sum_{\begin{array}{c}\substack{j=1,\\ j\neq i}\end{array}}^{N}\!\!\!P_U(j)H(P_j^*)+\sum_{l=1}^{|\mathcal{Y}|}P_U(i)\lambda_lH(\psi_l^*),
\end{align*}
which is contradiction with the assumption that $P^*_1,...,P^*_N$ achieves the minimum of $H(Y|U)$. Thus, each $P^*_i$ has to belong to $\mathbb{S}_i^*$.
\section*{Appendix D}
Consider the set $\mathbb{S}=\{y\in\mathbb{R}^{|\mathcal{Y}|}|My=MP_Y\}$. Any element in $\mathbb{S}$ has sum elements equal to one since we have $M(y-P_Y)=0\Rightarrow y-P_Y\in \text{Null}(P_{X|Y})$ and from Lemma~\ref{null} we obtain $1^T(y-P_Y)=0\Rightarrow 1^Ty=1$. The basic solutions of $\mathbb{S}$ are $W_{\Omega}^*$ defined as follows: Let $\Omega = \{\omega_1,..,\omega_{|\mathcal{X}|}\}$, where $\omega_i\in\{1,..,|\mathcal{Y}|\}$, then
$
W_{\Omega}^*(\omega_i)=M_{\Omega}^{-1}MP_Y(i),
$
and other elements of $W_{\Omega}^*$ are zero. Thus, the sum over all elements of $M_{\Omega}^{-1}MP_Y$ is equal to one, since each element in $\mathbb{S}$ has to sum up to one. For the second statement consider the set $\mathbb{S}^{'}=\left\{y\in\mathbb{R}^{|\mathcal{Y}|}|My=MP_Y+\epsilon M\begin{bmatrix}
P_{X|Y_1}^{-1}J_u\\0
\end{bmatrix}\right\}$. As argued before, basic solutions of $\mathbb{S}^{'}$ are $V_{\Omega}^*$, where 
\begin{align}\label{s_u^*}
V_{\Omega}^*(\omega_i)= (M_{\Omega}^{-1}MP_Y+\epsilon M_{\Omega}^{-1}M\begin{bmatrix}
P_{X|Y_1}^{-1}J_u\\0\end{bmatrix})(i).
\end{align}
Here elements of $V_{\Omega}^*$ can be negative or non-negative. From Lemma~\ref{null2}, each element in $\mathbb{S}^{'}$ has to sum up to one. Thus, by using the first statement of this proposition, the sum over all elements of the vector $ M_{\Omega}^{-1}M\begin{bmatrix}
P_{X|Y_1}^{-1}J_u\\0\end{bmatrix}$ is equal to zero.
	\section*{Appendix E}\label{appae}
	In the following we use the notation $\cong$ which is defined as 
	\begin{align*}
	f(x) = g(x) + o(x) \rightarrow f(x)\cong g(x),
	\end{align*}
	where $g(x)$ is the first Taylor expansion of $f(x)$.
	\begin{lemma}\label{5}
		Assume $P_{XY}\in\mathcal{H}_{XY}$ and $V_{\Omega_u}^*$ is an extreme point of the set $\mathbb{S}_u$, then for $P_{Y|U=u}=V_{\Omega_u}^*$ we have
		\begin{align*}
		H(P_{Y|U=u}) &=\sum_{y=1}^{|\mathcal{Y}|}-P_{Y|U=u}(y)\log(P_{Y|U=u}(y))\\&=-(b_u+\epsilon a_uJ_u)+o(\epsilon),
		\end{align*}
		with $b_u = l_u \left(M_{\Omega_u}^{-1}MP_Y\right),\ 
		a_u = l_u\left(M_{\Omega_u}^{-1}M(1\!\!:\!\!|\mathcal{X}|)P_{X|Y_1}^{-1}\right)\in\mathbb{R}^{1\times|\mathcal{X}|},\
		l_u = 
		\left[\log\left(M_{\Omega_u}^{-1}MP_{Y}(i)\right)\right]_{i=1:|\mathcal{X}|}\in\mathbb{R}^{1\times|\mathcal{X}|},
		$ and $M_{\Omega_u}^{-1}MP_{Y}(i)$ stands for $i$-th ($1\leq i\leq |\mathcal{X}|$) element of the vector $M_{\Omega_u}^{-1}MP_{Y}$. Furthermore, $M(1\!\!:\!\!|\mathcal{X}|)$ stands for submatrix of $M$ with first $|\mathcal{X}|$ columns. 
	\end{lemma} 
\begin{proof}
	Let $V_{\Omega_u}^*$ be defined as in \eqref{defin} and let $V_{|\mathcal{X}|\times1}=M_{\Omega_u}^{-1}MP_Y(i)+\epsilon M_{\Omega_u}^{-1}M(1:|\mathcal{X}|)P_{X|Y_1}^{-1}J_u$, thus we have
	\begin{align*}
	&H(P_{Y|U=u})=\sum_i V(i)\log(V(i))=-\sum_iV(i)\times\\  &\left(\log(M_{\Omega_u}^{-1}MP_Y(i))\!+\!\log(1\!+\!\epsilon\frac{M_{\Omega_u}^{-1}M(1\!:\!|\mathcal{X}|)P_{X|Y_1}^{-1}J_u(i)}{M_{\Omega_u}^{-1}MP_Y(i)})\right)\\
	&\stackrel{(a)}{=}-\sum_i V(i)\times\\&\left(\log(M_{\Omega_u}^{-1}MP_Y(i))+\epsilon\frac{M_{\Omega_u}^{-1}M(1\!:\!|\mathcal{X}|)P_{X|Y_1}^{-1}J_u(i)}{M_{\Omega_u}^{-1}MP_Y(i)} \right)\\
	&\stackrel{(b)}{=}-\left[\sum_i \log(M_{\Omega_u}^{-1}MP_Y(i))M_{\Omega_u}^{-1}MP_Y(i)\right]\\&-\!\epsilon\!\!\left[\sum_i M_{\Omega_u}^{-1}M(1\!:\!|\mathcal{X}|)P_{X|Y_1}^{-1}J_u(i)\log(M_{\Omega_u}^{-1}MP_Y(i)) \right]\!\!+\!o(\epsilon)\\&=-(a_u+\epsilon b_uJ_u)+o(\epsilon),
	\end{align*}
	where $a_u$ and $b_u$ are defined in Lemma~\ref{5}. Furthermore, in (a) we used the first order Taylor expansion of $\log(1+x)$ and in (b) we used the fact that $\sum_i M_{\Omega_u}^{-1}M(1\!:\!|\mathcal{X}|)P_{X|Y_1}^{-1}J_u(i)=0$ which can be proved as follows.
	Let $K_u=M_{\Omega_u}^{-1}M(1\!:\!|\mathcal{X}|)P_{X|Y_1}^{-1}J_u(i)$ and let $\Omega_u=\{\omega_1,..,\omega_{|\mathcal{X}|}\}$. For $i\in\{1,..,|\mathcal{X}|\}$, we define $K_u'\in{\mathbb{R}^{|\mathcal{Y}|}}$ as 
	$
	K_u'(\omega_i) = K_u(i), 
	$
	and other elements are set to be zero. We want to show that $1^TK_u=0$. It can be seen that  
	$
	MK_u'=M_{\Omega_u}K_u=M(1\!:\!|\mathcal{X}|)P_{X|Y_1}^{-1}J_u$
	which implies
	$M(K_u'-\begin{bmatrix}
	P_{X|Y_1}^{-1}J_u\\0
	\end{bmatrix})=0.
	$
	Thus, by using Lemma~\ref{null}, we obtain 
	$
	1^TK_u=1^TK_u'=1^T\begin{bmatrix}
	P_{X|Y_1}^{-1}J_u\\0
	\end{bmatrix}=1^TJ_u=0.
	$
	\end{proof}
\section*{Appendix F}
First, we find the sets $\mathbb{S}_u^*$ for $u\in\{1,2,3,4\}$. By using the SVD of the leakage matrix we have
\begin{align*}
M = \begin{bmatrix}
-0.5 \ -.5 \ -0.5 \ -0.5\\0.5345\ -0.8018\ 0\ 0.2673
\end{bmatrix}.
\end{align*}
Since $|\mathcal{X}|=2$, possible sets of $\Omega$ are $\Omega_{u_1}=\{1,2\}$, $\Omega_{u_2}=\{1,3\}$, $\Omega_{u_3}=\{1,4\}$, $\Omega_{u_4}=\{2,4\}$, $\Omega_{u_5}=\{2,3\}$ and $\Omega_{u_6}=\{3,4\}$. For these sets we calculate $M_{\Omega_{u_i}}^{-1}MP_Y$ as follows
\begin{align*}
&M_{\Omega_{u_1}}^{-1}MP_Y =\begin{bmatrix}
0.675\\0.375
\end{bmatrix},\ M_{\Omega_{u_2}}^{-1}MP_Y =\begin{bmatrix}
0.1875\\0.8125
\end{bmatrix}, \\
&M_{\Omega_{u_3}}^{-1}MP_Y =\begin{bmatrix}
-0.625\\1.625
\end{bmatrix},\ M_{\Omega_{u_4}}^{-1}MP_Y =\begin{bmatrix}
-0.125\\1.125
\end{bmatrix},\\
&M_{\Omega_{u_5}}^{-1}MP_Y =\begin{bmatrix}
0.1563\\0.8437
\end{bmatrix},\ M_{\Omega_{u_6}}^{-1}MP_Y =\begin{bmatrix}
0.625\\0.375
\end{bmatrix}. 
\end{align*}
The sets $\Omega_{u_3}$ and $\Omega_{u_4}$ produce negative elements and thus we only consider the sets $\Omega_{u_1}$, $\Omega_{u_2}$, $\Omega_{u_5}$ and $\Omega_{u_6}$ to construct the extreme points of $\mathbb{S}_u$. Let $J_u=\begin{bmatrix}
J_u^1\\J_u^2
\end{bmatrix}$ for $u\in\{1,2,3,4\}$. By using \eqref{prop1} $J_u^1+J_u^2=0$ we can show $J_u$ by $\begin{bmatrix}
-J_u^2\\J_u^2
\end{bmatrix}$. The sets $\mathcal{S}_u^*$ for $u\in\{1,2,3,4\}$ are obtained as (using \eqref{defin})
\begin{align*}
&V_{\Omega_{u_1}}^* = \begin{bmatrix}
0.675+\epsilon2J_u^2\\0.325-\epsilon2J_u^2\\0\\0
\end{bmatrix},\
V_{\Omega_{u_2}}^* = \begin{bmatrix}
0.1875+\epsilon 5J_u^2\\0\\0.8125-\epsilon 5J_u^2\\0
\end{bmatrix}\\
& V_{\Omega_{u_5}}^* = \begin{bmatrix}
0\\0.1563-\epsilon 2.5J_u^2\\0\\0.8437+\epsilon 2.5J_u^2
\end{bmatrix},\
V_{\Omega_{u_6}}^* = \begin{bmatrix}
0\\0\\0.6251-\epsilon 10J_u^2\\0\\0.3749+\epsilon 10J_u^2
\end{bmatrix},
\end{align*}
thus, we have $\mathbb{S}_u^*=\{V_{\Omega_{u_1}}^*,V_{\Omega_{u_2}}^*,V_{\Omega_{u_5}}^*,V_{\Omega_{u_6}}^*\}$ for $u\in\{1,2,3,4\}$. Since for each $u$, $|\mathbb{S}_u^*|=4$, $4^4$ optimization problems need to be considered while not all are feasible. In Step 2, we choose first element of $\mathbb{S}_1^*$, second element of $\mathbb{S}_2^*$, third element of $\mathbb{S}_3^*$ and fourth element of $\mathbb{S}_4^*$. Thus we have the following quadratic problem
\begin{align*}
\min\ &P_10.9097+P_20.6962+P_30.6254+P_40.9544\\-&P_1\epsilon J_1^2 2.1089+P_2\epsilon J_2^2 10.5816-P_3 \epsilon J_3^2 6.0808 \\+&P_4\epsilon J_4^2 7.3747\\
&\text{s.t.}\begin{bmatrix}
\frac{1}{2}\\\frac{1}{4}\\ \frac{1}{8}\\ \frac{1}{8}
\end{bmatrix} = P_1 \begin{bmatrix}
0.675+\epsilon2J_1^2\\0.325-\epsilon2J_1^2\\0\\0
\end{bmatrix}+ P_2\begin{bmatrix}
0.1875+\epsilon 5J_2^2\\0\\0.8125-\epsilon 5J_2^2\\0
\end{bmatrix}\\&+P_3\begin{bmatrix}
0\\0.1563-\epsilon 2.5J_3^2\\0\\0.8437+\epsilon 2.5J_3^2
\end{bmatrix}+P_4\begin{bmatrix}
0\\0\\0.6251-\epsilon 10J_4^2\\0.3749+\epsilon 10J_4^2
\end{bmatrix},\\
&P_1J_1^2+P_2J_2^2+P_3J_3^2+P_4J_4^2=0,\ P_1,P_2,P_3,P_4\geq 0,\\
&|J_1^2|\leq \frac{1}{2},\ |J_2^2|\leq \frac{1}{2},\ |J_3^2|\leq \frac{1}{2},\ |J_4^2|\leq \frac{1}{2},
\end{align*}  
where the minimization is over $P_u$ and $J_u^2$ for $u\in\{1,2,3,4\}$. Now we convert the problem to a linear program. We have
\begin{align*}
\eta_1&= \begin{bmatrix}
0.675P_1+\epsilon2P_1J_1^2\\0.325P_1-\epsilon2P_1J_1^2
\end{bmatrix}=\begin{bmatrix}
\eta_1^1\\\eta_1^2
\end{bmatrix},\\
\eta_2&=\begin{bmatrix}
0.1875P_2+\epsilon 5P_2J_2^2\\0.8125P_2-\epsilon 5P_2J_2^2
\end{bmatrix}=\begin{bmatrix}
\eta_2^1\\\eta_2^2
\end{bmatrix},\\
\eta_3&= \begin{bmatrix}
0.1563P_3-\epsilon 2.5P_3J_3^2\\0.8437P_3+\epsilon 2.5P_3J_3^2
\end{bmatrix}=\begin{bmatrix}
\eta_3^1\\\eta_3^2
\end{bmatrix},\\
\eta_4&= \begin{bmatrix}
0.3749P_4-\epsilon 10P_4J_4^2\\0.6251P_4+\epsilon 10P_4J_4^2
\end{bmatrix}=\begin{bmatrix}
\eta_3^1\\\eta_3^2
\end{bmatrix}.
\end{align*}
\begin{align*}
\min\ &0.567\eta_1^1+1.6215\eta_1^2+2.415\eta_2^1+0.2995\eta_2^2\\
&2.6776\eta_3^1+0.2452\eta_3^2+0.6779\eta_4^1+1.4155\eta_4^2\\
&\text{s.t.}\begin{bmatrix}
\frac{1}{2}\\\frac{1}{4}\\ \frac{1}{8}\\ \frac{1}{8}
\end{bmatrix} = \begin{bmatrix}
\eta_1^1+\eta_2^1\\\eta_1^2+\eta_3^2\\\eta_2^2+\eta_4^1\\\eta_3^2+\eta_4^2
\end{bmatrix},\ \begin{cases}
&\eta_1^1+\eta_1^2\geq 0\\
&\eta_2^1+\eta_2^2\geq 0\\
&\eta_3^1+\eta_3^2\geq 0\\
&\eta_4^1+\eta_4^2\geq 0
\end{cases}\\
&\frac{0.325\eta_1^1-0.675\eta_1^2}{2}+\frac{0.8125\eta_2^1-0.1875\eta_2^2}{5}+\\
&\frac{0.1563\eta_3^2-0.8437\eta_3^1}{2.5}+\frac{0.6251\eta_4^2-0.3749\eta_4^1}{10}=0,\\
&\frac{|0.325\eta_1^1-0.675\eta_1^2|}{\eta_1^1+\eta_1^2}\!\leq \epsilon,\ \!\frac{|0.8125\eta_2^1-0.1875\eta_2^2|}{\eta_2^1+\eta_2^2}\!\leq 2.5\epsilon\\
&\frac{|0.1563\eta_3^2-0.8437\eta_3^1|}{\eta_3^1+\eta_3^2}\leq 1.125\epsilon,\\ &\frac{|0.6251\eta_4^2-0.3749\eta_4^1|}{\eta_4^1+\eta_4^2}\leq 5\epsilon.
\end{align*}

The solution to the obtained problem is as follows (we assumed $\epsilon = 10^{-2}$)
\begin{align*}
&P_U = \begin{bmatrix}
0.7048 \\ 0.1492 \\ 0.146 \\ 0
\end{bmatrix},\
J_1 = \begin{bmatrix}
-0.0023 \\ 0.0023
\end{bmatrix},\ J_2 = \begin{bmatrix}
0.5 \\ -0.5
\end{bmatrix}\\
&J_3 = \begin{bmatrix}
-0.5 \\ 0.5
\end{bmatrix},\
J_4 = \begin{bmatrix}
0 \\ 0
\end{bmatrix},\ \min \text{(cost)} = 0.8239,
\end{align*}
thus, for this combination we obtain $I(U;Y)\cong H(Y)-0.8239=1.75-0.8239 = 0.9261$. If we try all possible combinations we get the minimum cost as $0.8239$ so that we have $\max I(U;Y)\cong 0.9261$. For feasibility of $\sum_{u=1}^{|\mathcal{Y}|} P_uV_{\Omega_u}^*=P_Y$, consider the following combination: We choose first extreme point from each set $\mathbb{S}_u$ for $u\in\{1,2,3,4\}$. Considering the condition $\sum_{u=1}^{|\mathcal{Y}|} P_uV_{\Omega_u}^*=P_Y$, we have
\begin{align*}
\begin{bmatrix}
\frac{1}{2}\\\frac{1}{4}\\ \frac{1}{8}\\ \frac{1}{8}
\end{bmatrix} = P_1 \begin{bmatrix}
0.675+\epsilon2J_1^2\\0.325-\epsilon2J_1^2\\0\\0
\end{bmatrix}+ P_2\begin{bmatrix}
0.675+\epsilon2J_2^2\\0.325-\epsilon2J_2^2\\0\\0
\end{bmatrix}
\end{align*}
\begin{align*}+P_3\begin{bmatrix}
0.675+\epsilon2J_3^2\\0.325-\epsilon2J_3^2\\0\\0
\end{bmatrix}+P_4\begin{bmatrix}
0.675+\epsilon2J_4^2\\0.325-\epsilon2J_4^2\\0\\0
\end{bmatrix},
\end{align*} 
which is not feasible.
\section*{Appendix G}
Let $P_{XY}\in \mathcal{H}_{XY}$ and let $\Omega^1$ be the set of all $\Omega_i\in\{1,..,|\mathcal{Y}|\},\ |\Omega|=|\cal X|$, such that each $\Omega_i$ produces a valid standard distribution vector $M_{\Omega_i}^{-1}MP_Y$, i.e., all elements in the vector $M_{\Omega_i}^{-1}MP_Y$ are positive. Let $\Omega^2$ be the set of all $\Omega_j\in\{1,..,|\mathcal{Y}|\},\ |\Omega|=|\cal X|$, such that for each $\Omega_j$, the vector $M_{\Omega_j}^{-1}MP_Y$ contains at least one negative element. Clearly, since $P_{XY}\in \mathcal{H}_{XY}$, $\Omega^1\cup \Omega^2$ are all subsets of $\{1,..,|\cal Y|\}$ with cardinality $|\cal X|$. In Theorem~2 the first order Taylor expansion of $\log(1+x)$ in order to approximate the conditional entropy $H(Y|U)$ is used. Therefore, we must have $|\epsilon\frac{M_{\Omega}^{-1}M(1:|\mathcal{X}|)P_{X|Y_1}^{-1}J_u(y)}{M_{\Omega}^{-1}MP_Y(y)}|<1$ for all $y\in\{1,..,|\cal X|\}$ and $\Omega\in\Omega^1$. One sufficient condition for $\epsilon$ to satisfy this inequality is to have $\epsilon< \frac{\min_{y,\Omega\in \Omega^1} M_{\Omega}^{-1}MP_Y(y)}{\max_{\Omega\in \Omega^1} |\sigma_{\max} (H_{\Omega})|}$. In this case for all $y$ and $\Omega\in\Omega^1$ we have
\begin{align*}
\epsilon^2|H_{\Omega}J_u(y)|^2&\leq \epsilon^2\sigma_{\max}^2(H_{\Omega})\lVert J_u\rVert_2^2 \stackrel{(a)}{\leq} \epsilon^2 \max_{\Omega\in \Omega^1} \sigma^2_{\max} (H_{\Omega})\\ &< (\min_{y,\Omega\in \Omega^1} M_{\Omega}^{-1}MP_Y(y))^2\leq (M_{\Omega}^{-1}MP_Y(y))^2,
\end{align*}
which implies
\begin{align*}
\epsilon|M_{\Omega}^{-1}M(1:|\mathcal{X}|)P_{X|Y_1}^{-1}J_u(y)|< M_{\Omega}^{-1}MP_Y(y).
\end{align*}
The step (a) follows from $\lVert J_u\rVert_2^2\leq\lVert J_u\rVert_1^2\leq 1$. Note that $H_{\Omega}$ is an invertible matrix so that $|\sigma_{\max}|>0$.
Furthermore, we should ensure that not all negative elements of the vectors $M_{\Omega}^{-1}MP_Y$ for $\Omega\in \Omega^2$ can become positive by adding $M_{\Omega}^{-1}M(1:|\mathcal{X}|)P_{X|Y_1}^{-1}J_u$. One sufficient condition for $\epsilon$ is to have $\epsilon< \frac{\min_{\Omega\in \Omega^2} \max_{y:M_{\Omega}^{-1}MP_Y(y)<0} |M_{\Omega}^{-1}MP_Y(y)|}{\max_{\Omega\in \Omega^2} |\sigma_{\max} (H_{\Omega})|}$, since in this case for all $y$ and $\Omega\in\Omega^2$ we have
\begin{align*}
\epsilon^2|H_{\Omega}J_u(y)|^2&\leq \epsilon^2\sigma_{\max}^2(H_{\Omega})\lVert J_u\rVert_2^2 \leq \epsilon^2 \max_{\Omega\in \Omega^2} \sigma^2_{\max} (H_{\Omega})\\ &< \left(\min_{\Omega\in \Omega^2} \max_{y:M_{\Omega}^{-1}MP_Y(y)<0} |M_{\Omega}^{-1}MP_Y(y)|\right)^2,
\end{align*}
which implies
\begin{align*}
\epsilon|M_{\Omega}^{-1}M(1:|\mathcal{X}|)P_{X|Y_1}^{-1}J_u(y)|< \!\!\!\!\!\!\max_{\begin{array}{c}\substack{y:\\M_{\Omega}^{-1}MP_Y(y)<0}\end{array}}\!\!\!\!\!\! |M_{\Omega}^{-1}MP_Y(y)|,
\end{align*}
for all $y$ and $\Omega\in\Omega^2$.
Last inequality ensures that for all $\Omega\in \Omega^2$ the vector $M_{\Omega}^{-1}MP_Y+\epsilon M_{\Omega}^{-1}M\begin{bmatrix}
P_{X|Y_1}^{-1}J_u\\0\end{bmatrix}$ contains at least one negative element which makes it an infeasible vector distribution. 
Thus, an upper bound for a valid $\epsilon$ in our work is as follows
$\epsilon < \min\{\epsilon_1,\epsilon_2\}$.
In summary the upper bound $\epsilon_2$ ensures that none of the sets $\Omega_j\in\Omega^2$ can produce basic feasible solution of $\mathbb{S}_u$ and the upper bound $\epsilon_1$ ensures that the approximation used in this paper is valid.
\section*{Appendix H}
For the first part we have
\begin{align*}
\text{MMSE}(X|U)&=E(X^2)-\sum_u P_u[\sum_x xP_{x|u}]^2\\&=E(X^2)-\sum_u P_u[\sum_x xP_x+\epsilon\sum xJ_u(x)]^2\\
&\stackrel{(a)}{=}\text{Var}(X)-\epsilon^2\sum_u P_u(\sum xJ_u)^2\\&\stackrel{(b)}{=}\text{Var}(X)-\epsilon^2\sum_u P_u(J_u(1))^2(x_1-x_2)^2\\&\stackrel{(c)}{\geq} \text{Var}(X)-\frac{1}{4}\epsilon^2(x_1-x_2)^2,
\end{align*}
where (a) follows from $\mathbb{E}(X)=0$, (b) follows from $J_u(1)+J_u(2)=0$ and (c) comes from the fact that $|J_u(1)|\leq \frac{1}{2}$.
For second part we have
\begin{align*}
\text{Var}(X)&=p_{x_1}x_1^2+p_{x_2}x_2^2=p_{x_1}(x_1^2-x_2^2)+x_2^2\\&\stackrel{(a)}{=}-x_2(x_1+x_2)+x_2^2=-x_1x_2\leq \frac{1}{4}(x_1-x_2)^2,
\end{align*}
and the equality holds for $x_1=-x_2$. Step (a) follows from $p_{x_1}(x_1-x_2)=-x_2$.
\bibliographystyle{IEEEtran}
\bibliography{IEEEabrv,Amir}
\end{document}